\documentclass[12pt]{article}
\usepackage{geometry}
\usepackage{amsmath, amsthm, amsfonts, amssymb, bbm, bm, amsbsy}
\usepackage{enumitem}
\usepackage{graphicx}
\usepackage{xypic}
\usepackage[all]{xy}
\usepackage[english]{babel}
\usepackage[font=small,format=plain,labelfont=bf,up]{caption}
\usepackage[bottom]{footmisc}

\usepackage{hyperref}       
\hypersetup{       
   pdftex,                  
   colorlinks=true,        
   pdfstartview=FitH,      
   allcolors=[rgb]{0,0.1,0.5},        
   bookmarks=true,          
	 bookmarksdepth=section, 
   bookmarksopen=false,     
   bookmarksnumbered=true, 
	 backref, pagebackref  ,
   linktocpage=true
}

\theoremstyle{plain}

\newtheorem{theorem}{Theorem}
\newtheorem{lemma}[theorem]{Lemma}
\newtheorem{proposition}[theorem]{Proposition}

\newtheorem{corollary}[theorem]{Corollary}

\theoremstyle{definition}

\newtheorem{remark}[theorem]{Remark}
\newtheorem{example}[theorem]{Example}
\newtheorem{definition}[theorem]{Definition}
\newtheorem{notation}[theorem]{Notations}

 \numberwithin{equation}{section}
 \numberwithin{theorem}{section}

\def\be#1\ee{\begin{equation}#1\end{equation}}
\def\ba#1\ea{\begin{align}#1\end{align}}

\newcommand\eps           {\varepsilon}

\newcommand\Rb            {\mathbb{R}}
\newcommand\Zb            {\mathbb{Z}}

\newcommand{\void}[1]{}

\renewcommand{\vec}[1]{\mathbf{#1}}

\begin{document}

\thispagestyle{empty}
\def\thefootnote{\fnsymbol{footnote}}
\begin{flushright}
ZMP-HH/17-23\\
Hamburger Beitr\"age zur Mathematik 678
\end{flushright}
\vskip 3em
\begin{center}\LARGE
Existence and uniqueness of solutions 
\\
to Y-systems and TBA equations
\end{center}

\vskip 2em
\begin{center}
{\large 
Lorenz Hilfiker~~and~~Ingo Runkel~\footnote{Emails: {\tt lorenz.hilfiker@uni-hamburg.de}~,~{\tt ingo.runkel@uni-hamburg.de}}}
\\[1em]
Fachbereich Mathematik, Universit\"at Hamburg\\
Bundesstra\ss e 55, 20146 Hamburg, Germany
\end{center}

\vskip 2em
\begin{center}
  July 2017
\end{center}
\vskip 2em

\begin{abstract}
We consider Y-system functional equations of the form 
$$ 
\textstyle{
Y_n(x+i)Y_n(x-i)=\prod_{m=1}^N \left(1+Y_m(x)\right)^{G_{nm}}
}
$$ 
and the corresponding nonlinear integral equations of the Thermodynamic Bethe Ansatz.
We prove an existence and uniqueness result for solutions of these equations, subject to appropriate conditions on the analytical properties of the $Y_n$, in particular the absence of zeros in a strip around the real axis. The matrix $G_{nm}$ 
must have non-negative real entries, and be irreducible and diagonalisable over $\mathbb{R}$ with spectral radius less than 2. This includes the adjacency matrices of finite Dynkin diagrams, but covers much more as we do not require 
	$G_{nm}$ 
to be integers. Our results specialise to the constant Y-system, proving existence and uniqueness of a strictly positive solution in that case.
\end{abstract}

\setcounter{footnote}{0}
\def\thefootnote{\arabic{footnote}}

\newpage

{\small

\tableofcontents

}

\newpage

\section{Introduction and summary}

In this paper we investigate the existence and uniqueness of solutions to two related sets of equations. The first set consists of algebraic equations for $N$ analytic functions $Y_n$, and is an example for so-called Y-system functional equations. The second set consists of coupled
nonlinear integral equations for $N$ functions $f_n$, called thermodynamic Bethe ansatz equations, or TBA equations for short.
Y-systems and their relation to 
	TBA equations 
first appeared in \cite{Zamo:ADE}.

We will now describe these two sets of equations in more detail and  state our existence and uniqueness results. 
Afterwards we outline the standard argument connecting the Y-system to the TBA equations and stress the points where our approach differs from earlier works.

\subsubsection*{Main results}

Let us start by fixing some notation which we need to formulate the results and which we will use throughout this paper.

\begin{notation}\label{intro-notation}~
\begin{itemize}[leftmargin=1em]
\item Let $\mathbb{K}$ stand for $\mathbb{R}$ or $\mathbb{C}$. 
We denote by \[BC(\mathbb{R},\mathbb{K})\] the functions from $\mathbb{R}$ to $\mathbb{K}$ which are continuous and bounded, and by
$$
	BC_-(\mathbb{R},\mathbb{R})  
$$
the continuous real-valued functions on $\mathbb{R}$ which are bounded from below.
\item For $a>0$ we denote by $\mathbb{S}_a := \lbrace z\in\mathbb{C} | -a<\mathrm{Im}(z)<a\rbrace$ the open horizontal strip in $\mathbb{C}$ of height $2a$, and by $\overline{\mathbb{S}}_a$ its closure.
We define the spaces of functions
$$
	\mathcal{A}(\mathbb{S}_a) ~\supset~ \mathcal{BA}(\mathbb{S}_a) \ ,
$$
where $\mathcal{A}(\mathbb{S}_a)$ is the space of $\mathbb{C}$-valued functions which are analytic on $\mathbb{S}_a$ and have a continuous extension to 
	$\overline{\mathbb{S}}_a$, 
and $\mathcal{BA}(\mathbb{S}_a)$ consists of 
	those functions in $\mathcal{A}(\mathbb{S}_a)$
which are in addition bounded on 
	$\overline{\mathbb{S}}_a$.
\item
We fix the constants
$$
	N \in \mathbb{Z}_{>0}
	\quad ,
	\quad
	s \in \mathbb{R}_{>0} \ .
$$
\item We denote by 
$$
	\mathrm{Mat}_{<2}(N) \subset \mathrm{Mat}(N,\mathbb{R})
$$	
the subset of real-valued $N\times N$
matrices which can be diagonalised over the real numbers, and whose eigenvalues lie in the open interval $(-2,2)$.
\end{itemize}
\end{notation}

We can now describe the Y-system. For $\vec{G} \in \mathrm{Mat}(N,\mathbb{R})$ (with entries $G_{nm}$) and 
	$Y_1,\dots,Y_N \in \mathcal{A}(\mathbb{S}_s)$, 
the Y-system is the set of functional equations
\be\label{Y}
Y_n(x+is)Y_n(x-is)=\prod_{m=1}^N \left(1+Y_m(x)\right)^{G_{nm}}  \ ,\tag{Y} 
\ee
for $n=1,\dots,N$ and all $x \in \mathbb{R}$.
If $\vec{G}$ is not integer valued, one needs to give a prescription how to deal with the multi-valuedness of the right hand side. We will later do this by demanding $Y_n$ to be positive and real-valued on the real axis.

If $Y_n\in\mathcal{A}(\mathbb{S}_s)$ has no zeros,
we may pick $h_n\in\mathcal{A}(\mathbb{S}_s)$ such that $Y_n(z)=e^{h_n(z)}$ for all 
	$z\in\overline{\mathbb{S}}_s$. 
Denote $h_n(z)=\log Y_n(z)$. 
	We can think of a function $a_n\in\mathcal{A}(\mathbb{S}_s)$ as capturing the asymptotic behaviour of $Y_n(z)$ if $\log Y_n(z)-a_n(z)$ is bounded on 
	$\overline{\mathbb{S}}_s$. 
This condition is independent of the branch choice
	for the logarithm.
To formulate our first main theorem, we need to single out a certain type of asymptotic behaviour.

\begin{definition}
We call $\vec{a}\in\mathcal{A}(\mathbb{S}_s)^N$ (with components $a_n(z)$) 
	a {\sl valid asymptotics for \eqref{Y}} if
\begin{enumerate}
\item for $n=1,...,N\, $, $a_n$ is real valued on $\mathbb{R}$ and the functions $e^{-a_n(x)}$ and $\frac{d}{dx}e^{-a_n(x)}$ are bounded on $\mathbb{R}$,
\item $\vec a$ satisfies
\be\label{linear-eqn-for-asympt}
\vec{a}(x+is)+\vec{a}(x-is) = \vec{G}\cdot \vec{a}(x)
 \qquad \text{ for all } x \in \mathbb{R} \ .
\ee
\end{enumerate}
\end{definition} 

Recall from Perron-Frobenius theory that
a real $N\times N$ matrix is called {\sl non-negative} if all its entries are $\ge 0$, and {\sl irreducible} if there is no permutation of the standard basis
which makes it block-upper-triangular.
Our first main result is the following existence and uniqueness statement.

\begin{theorem}\label{main-theorem-Y}
Let $\vec{G}\in\mathrm{Mat}_{<2}(N)$ be non-negative and irreducible, and $\vec{a}\in\mathcal{A}(\mathbb{S}_s)^N$ a valid asymptotics for \eqref{Y}.
Then there exists a solution
$Y_1,\dots,Y_N \in \mathcal{A}(\mathbb{S}_s)$ to \eqref{Y} which satisfies, for $n=1,\dots,N$ ,
\begin{enumerate}
\item \label{Yproperties:real}
$Y_n(\mathbb{R})\subseteq\mathbb{R}_{>0}$~,
 \hfill\emph{(real\,\&\,positive)}
\item \label{Yproperties:roots} $Y_n(z)\neq 0$ for all $z\in\overline{\mathbb{S}}_s$~.
\hfill\emph{(no roots)}
\item \label{Yproperties:asymptotics}
$ \log Y_n(z)-a_n(z) \in\mathcal{BA}(\mathbb{S}_s)$.\hfill\emph{(asymptotics)}
\end{enumerate}
Moreover, this solution is the unique one in $\mathcal{A}(\mathbb{S}_s)$ which satisfies properties \ref{Yproperties:real}--\ref{Yproperties:asymptotics}.
\end{theorem}

	Recall that the logarithm in property~\ref{Yproperties:asymptotics} exists on all of $\overline{\mathbb{S}}_s$ as by condition \ref{Yproperties:roots}, $Y_n$ has no zeros, and that property~\ref{Yproperties:asymptotics} is not affected by the choice of branch for the logarithm.

Even though the unique solution $Y_1,\dots,Y_N$ is initially only defined on $\overline{\mathbb{S}}_s$, using \eqref{Y} and property~\ref{Yproperties:roots}, it is easy to see that $Y_n$ can be analytically continued at least to $\overline{\mathbb{S}}_{3s}$.

\medskip

	One important valid asymptotics for \eqref{Y} is simply $\vec{a}=0$,
 in which case the $Y_n$ themselves are bounded. We will see in Corollary~\ref{unique-constant-sol} below that then in fact the $Y_n$ are constant. 
The Perron-Frobenius eigenvector $\vec{w}$ of $\vec{G}$ provides a whole family of valid asymptotics.
By our assumptions on $\vec G$, $\vec{w}$ can be chosen to have positive entries and its eigenvalue lies strictly between $0$ and $2$
(see Theorem \ref{PF}).
Then for any choice of $\gamma \in \mathbb{R}_{>0}$ such that $\vec{G}\cdot\vec{w} = 2 \cos(\gamma) \vec{w}$, the functions
\be \label{asymptoticexample}
	\vec{a}(x) = e^{\gamma x/s} \, \vec{w} 
	\qquad \text{and} \qquad
	\vec{a}(x) = e^{- \gamma x/s} \, \vec{w} 
\ee
	are valid asymptotics for \eqref{Y}.
We can also take linear combinations with positive coefficients; in particular the symmetric choice
$\vec{a}(x) = r \cosh(\gamma x/s) \, \vec{w}$ 
is considered frequently in the context of massive relativistic quantum field theory, where $\vec{w}$ takes the role of the mass vector
and $r>0$ represents the volume.

\medskip

Next we discuss the TBA equations. Let $\vec{C} \in \mathrm{Mat}_{<2}(N)$ and consider the following Fourier transform of a matrix-valued function,
	for $x \in \mathbb{R}$,
\be\label{phiC-def-intro}
\Phi_{\vec{C}}(x)~=~\frac{1}{2\pi}\int_{-\infty}^{\infty} 
\hspace{-.4em}
e^{ikx}
\, \big( 2\cosh(sk)\mathbf{1}-\vec{C} \big)^{-1} \,dk \ .
\ee
The integral is well defined since by the condition on the eigenvalues of $\vec{C}$, the components of the integrand are Schwartz-functions.
Then the 
components of $\Phi_{\vec{C}}$ are also Schwartz-functions which moreover are real and even.
See Section~\ref{section-NdimGreen} for more details on $\Phi_{\vec{C}}$. 
For $\vec{a} \in BC_-(\Rb,\Rb)^N$, $\vec G \in \mathrm{Mat}(N,\mathbb{R})$, and $\vec{C}$ as above, the TBA equation is the following nonlinear integral equation for a vector-valued function $\vec{f} \in BC(\mathbb{R},\mathbb{R})^N$:
\be\label{TBAintro}
\vec{f}(x)
~=\,
\int_{-\infty}^{\infty}
\hspace{-.5em}
\Phi_{\vec C}(x-y)\cdot\Big(\vec{G}\cdot\log\left(e^{-\vec{a}(y)}+e^{\vec{f}(y)}\right)-\vec{C}\cdot\vec{f}(y)\Big) \;dy \ . \tag{TBA}
\ee
Here we used the short-hand notation $\log\left(e^{-\vec{a}(y)}+e^{\vec{f}(y)}\right)$ to denote the function $\mathbb{R} \to \mathbb{R}^N$ with entries
\be\label{functions-of-a-vector-convention}
	\big[ \log\left(e^{-\vec{a}(y)}+e^{\vec{f}(y)}\right) \big]_j := 
	\log\left(e^{-a_j(y)}+e^{f_j(y)}\right)  \ .
\ee
The integral \eqref{TBAintro} is well-defined because the components of $\vec a$ are bounded from below, $\vec f$ is bounded and the components of $\Phi_{\vec{C}}$ are Schwartz-functions.

Recall that a function
	$f : \mathbb{R} \to \mathbb{K}$
is called \textsl{H\"older continuous} 
if there exist $0<\alpha\leq 1$ and $C>0$, such that 
\be 
	\left|f(x)-f(y) \right|\leq C \left|x-y\right|^\alpha \hspace{1cm}  \text{for all}~~ x,y\in\mathbb{R}\ .
\ee
If $\alpha=1$, then $f$ is called
 \textsl{Lipschitz continuous}.

Our second main result is:

\begin{theorem}\label{main-theorem-TBA} Let  $\vec{G} \in \mathrm{Mat}_{<2}(N)$ be non-negative and irreducible, $\vec{C} \in \mathrm{Mat}_{<2}(N)$, and $\vec{a}\in BC_-(\mathbb{R},\mathbb{R})^N$ such that the components of $e^{-\vec{a}}$ are H\"older continuous. 
Then the following holds:
\begin{enumerate}[label=\roman*)]
\item\label{mainTBAunique} The TBA equation \eqref{TBAintro} has a unique solution $\vec{f}_\star \in BC(\mathbb{R},\mathbb{R})^N$. The function $\vec{f}_\star$ is independent of the choice of $\vec{C}$.
\item\label{mainTBAsolvesY} $\vec{f}_\star$ can be extended to a function in 
$\mathcal{BA}(\mathbb{S}_s)^N$,
which we denote also by $\vec{f}_\star$.
If $\vec{a}$ can be extended to a valid asymptotics, then the functions
	$Y_n(z) = e^{a_n(z)+f_{\star,n}(z)}$, 
for $z \in \overline{\mathbb{S}}_s$ and $n=1,\dots,N$, provide the unique solution to \eqref{Y} with the properties as stated in Theorem~\ref{main-theorem-Y}.
\end{enumerate}
\end{theorem}

In the case $N=1$, $\vec{G}=\vec{C}=1$ and $\vec{a} = r\cosh(x)$, 
the existence and uniqueness of $\vec{f}_\star$ was already shown in \cite{FringKorffSchulz} (see discussion in Section~\ref{N=1-example-section}).
Theorems~\ref{main-theorem-Y} and~\ref{main-theorem-TBA}, as well as Corollary~\ref{unique-constant-sol} below, will be proved in Section~\ref{section:uniquenessY}.

\medskip

Next we specialise Theorems~\ref{main-theorem-Y} and~\ref{main-theorem-TBA} to the case $\vec{a}=0$. From the proofs of these theorems we get the following corollary. 

\begin{corollary}\label{unique-constant-sol}
For $\vec{a}=0$, the unique solution $\vec{f}_\star$ from Theorem~\ref{main-theorem-TBA} is constant, and correspondingly, the unique solution $Y_1,\dots,Y_N$ from Theorem~\ref{main-theorem-Y} is constant.
\end{corollary}

\begin{remark}\label{remark-constantY}
~
\begin{enumerate}[label=\roman*)] \setlength{\leftskip}{-1em}
\item
The constant case is interesting in itself. The functional equations \eqref{Y} turn into the {\sl constant Y-system}
\be\label{Yconst}
	Y_n^{\,2}
~=~\prod_{m=1}^N \left(1+Y_m\right)^{G_{nm}}  \ ,
\ee
for $Y_1,\dots,Y_N \in \mathbb{C}$.
Suppose $\vec{G} \in \mathrm{Mat}_{<2}(N)$ is non-negative and irreducible as in Theorem~\ref{main-theorem-Y}.
Since a real and positive solution to \eqref{Yconst} also solves \eqref{Y}
and satisfies conditions~\ref{Yproperties:real}--\ref{Yproperties:asymptotics} 
in Theorem~\ref{main-theorem-Y} (for $\vec a =0$),
by Corollary~\ref{unique-constant-sol} the constant Y-system 
\textsl{has a unique positive solution}.
This extends a result of \cite{Nahm:2009hf,Inoue:2010a}, 
where symmetric and positive definite $\vec{G}$ were considered, as well as adjacency matrices of finite Dynkin diagrams.
\item
If $\vec{G}$ is the adjacency matrix of a finite Dynkin diagram,
explicit trigonometric expressions for the unique positive solution to 
the constant Y-system (and more general versions thereof)
are known or conjectured, see \cite{Kirillov:1989} and \cite[Sec.\,14]{KunibaNakanishiSuzuki}.
\item\label{spectralradius2isbad}
If $\vec{G}$ has spectral radius
	$\ge 2$, 
it is shown in \cite[Sec.\,4]{Tateo:DynkinTBAs} that the constant Y-system
\eqref{Yconst} does not possess a real positive solution at all.
This shows that for $\vec a=0$, the condition on the spectral radius in Theorems~\ref{main-theorem-Y} and~\ref{main-theorem-TBA} is sharp.
\end{enumerate}
\end{remark}

\subsubsection*{Background on Y-systems and TBA equations}

The Thermodynamic Bethe Ansatz 
was developed to study thermodynamic properties of a gas of particles moving on a circle \cite{YangYang69}. The version for relativistic particles  whose scattering matrix factorises into a product of two-particle scattering matrices was given in \cite{Zamolodchikov:1989cf}. The reformulation as a Y-system was first described in \cite{Zamo:ADE}. A review of Y-systems and their applications can be found for example in \cite{KunibaNakanishiSuzuki}. Below we sketch the transformations linking the Y-system and the TBA equation, see \cite{Zamo:ADE} and e.g.\ \cite{Tateo:DynkinTBAs,Dorey:2007zx,vanTongeren:TBAreview}.

We note that while our proof follows the standard steps,
we are not aware of a previous proof of the correspondence between Y-systems and TBA equations in the literature, 
in the sense that all analytic questions are carefully addressed. To provide all these details was one of the motivations to write the present paper.

\medskip

Rewrite $Y_n$ in \eqref{Y} as $Y_n(z)=\exp(f_n(z)+a_n(z))$, where $f_n(z)$ are bounded functions and $a_n(z)$
are valid asymptotics for \eqref{Y}.
Upon taking the logarithm, one verifies that the $a_n(z)$ cancel out and one remains with the set of finite difference equations 
\be
f_n(x+is)+f_n(x-is)=\sum_{m=1}^N G_{nm}\log\left(e^{-a_m(x)}+e^{f_m(x)}\right)\ .
\ee 
Even though it looks like a trivial modification of the above equation, it will be crucial for us to add a linear term in $\vec{f}$ to both sides (we switch to vector notation to avoid too many indices, recall also the convention \eqref{functions-of-a-vector-convention}),
\be\label{intro-nonlinear-rhs}
\vec{f}(x+is)+\vec{f}(x-is) -\vec{C}\cdot\vec{f}(x) 
= 
\vec{G}\cdot\log\left(e^{-\vec{a}(y)}+e^{\vec{f}(y)}\right)-\vec{C}\cdot\vec{f}(x) 
 \ . 
\ee
To get rid of the nonlinearity for a while, we replace the $\vec f$-dependent function on the right hand side simply by a suitable function $\vec{g}$,
\be\label{intro-f+f-Cf=g}
\vec{f}(x+is)+\vec{f}(x-is) -\vec{C}\cdot\vec{f}(x) = \vec{g}(x) \ . 
\ee
The difference equation can now be solved by a Greens-function-like approach. Namely, the function $\Phi_{\vec C}$ from \eqref{phiC-def-intro} gives rise to a representation of the 
Dirac $\delta$-distribution (see Lemma~\ref{Phiproperties} for details):
\be
\Phi_{\vec C}(x+is)+\Phi_{\vec C}(x-is)-\vec{C}\cdot \Phi_{\vec C}(x) = \delta(x)\mathbf{1}_N \ .\ee
This allows one to write a solution to the functional equation \eqref{intro-f+f-Cf=g} as an integral,
\be\label{intro_convolution-to-solve-difference}
\vec{f}(x)=\int_{-\infty}^{\infty}\Phi_{\vec C}(x-y)\cdot\vec{g}(y) \;dy \ .
\ee
The only detailed proof of this that we are aware of is \cite[Lem.\,2]{TraceyWidom}, which treats the
case
$N=1$ and $\vec{C}=0$ and imposes 
	a decay condition on $\vec{f}(x)$ for $x\rightarrow\pm\infty$. 
	Therefore, in Section~\ref{section-soldefeqn} we give a proof in the generality we require.

Substituting the right hand side of \eqref{intro-nonlinear-rhs} for $\vec g$ in \eqref{intro_convolution-to-solve-difference}
produces \eqref{TBAintro}.

\begin{remark} \label{remark-S-matrix-relation}
In the case $\vec{C}=0$, the matrix $\Phi_{\vec C}(x)$ is
proportional to the identity matrix
and corresponds to the \textsl{standard kernel} 
(often denoted by $s$) which produces the
\textsl{universal} or \textsl{simplified TBA equations}
of the physics literature.
The case $\vec{C}=\vec{G}$ yields the canonical TBA equations which emphasise the relation to the relativistic
scattering matrix $\vec{S}(x)$ if such is available (see Remark~\ref{remark-Smatrix-exists}).
Specifically, we have (see e.g.\ \cite{Dorey:2007zx})
\be\label{PhiG-logS}
\left[\Phi_{\vec G}(x)\cdot \vec{G}\right]_{nm} = \frac{i}{2\pi}\frac{d}{dx}\log\left(S_{nm}(x)\right)\ ,
\ee
More details and references can be found e.g.\ in \cite{Zamo:ADE,Tateo:DynkinTBAs,vanTongeren:TBAreview}.
Note that our Green's function $\Phi_{\vec G}(x)$ has to be multiplied with $\vec{G}$ to obtain the canonical kernel used in the physics TBA-literature.
In this paper, we consistently treat the factor $\vec{G}$ not as part of the kernel, but absorb it in the function $\vec{g}(x)$. This is a natural convention for $\vec{C}=0$, and we preserve this convention for general $\vec{C}$.
\end{remark}

Allowing for arbitrary $\vec{C} \in \mathrm{Mat}_{<2}(N)$,
not just $\vec C = 0$ or $\vec G$, is one place where we work in greater generality 
than the physics literature we are aware of.
For us it will be important to choose $\vec{C} = \tfrac 12 \vec{G}$, as this will allow us to apply the Banach Fixed Point Theorem to find a unique solution to the integral equation (Proposition~\ref{TBAuniquenessDynkin}).

The other place where we allow for greater generality than considered before is in the choice of the matrix $\vec G$, which can have non-negative real entries, rather than 
	just
integers. 
In the non-negative integral case, the $\vec{G}$ which fit our assumptions include in particular the adjacency matrix of finite Dynkin diagrams  
or tadpole graphs ($T_N=A_{2N}/\mathbb{Z}_2$) -- these are called \textsl{Dynkin TBAs} in \cite{Tateo:DynkinTBAs}.

\medskip

Existence of a solution to equations similar to \eqref{TBAintro} has in some cases also been argued constructively. In \cite{YangYang69} and \cite{Lai} solutions to some specific TBA equations (albeit with $\Phi_{\vec{C}}(x)$ replaced by functions substantially different from ours) are constructed from a specific starting function by iterating the equations and showing uniform convergence.

A different approach to existence and uniqueness of solutions to the TBA equation \eqref{TBAintro} is suggested in a footnote in \cite[Sec.\,3.2]{KlassenMelzer91}, where the authors propose to use a fixed point theorem due to Leray, Schauder and Tychonoff. A detailed proof is, however, not provided.

Various methods to solve equations of the general form \eqref{TBAintro}, so called Hammerstein equations, are discussed in \cite{AppellChen}, including several fixed point theorems. In particular their example 1 is similar in spirit to \eqref{TBAintro}.

There are also other types of nonlinear integral equations relevant to the study of integrable models, which often share many features with TBA equations.
Existence and uniqueness of solutions to such an  equation of Destri-de-Vega type in the XXZ model has been investigated in \cite{Kozlowski}.

\subsubsection*{Structure of paper}

In Section~\ref{section-soldefeqn} we give a detailed statement and proof of the above claim that \eqref{intro_convolution-to-solve-difference} solves \eqref{intro-f+f-Cf=g}, see Proposition~\ref{FuncRelToNLIE}. In Section~\ref{section-unique-integral-soln} we apply the Banach Fixed Point Theorem to obtain a unique solution to \eqref{TBAintro} in the case $\vec C = \tfrac12 \vec G$, see Proposition~\ref{TBAuniquenessDynkin}. Section~\ref{section:uniquenessY} contains the 
	proofs of Theorems~\ref{main-theorem-Y}, \ref{main-theorem-TBA} and of Corollary~\ref{unique-constant-sol}. 
Finally, in Section~\ref{section-outlook} we discuss some open questions.

\subsubsection*{Acknowledgements}
We would like to thank
	Nathan Bowler,
	Andrea Cavagli\`a,
	Patrick Dorey,
	Clare Dunning,
	Andreas Fring,
	Annina Iseli,
	Karol Kozlowski,
	Louis-Hadrien Robert,
	Roberto Tateo,
	J\"org Teschner,
	Stijn van Tongeren,
	Alessandro Torrielli,
	Beno\^\i t Vicedo,
and
	G\'erard Watts
for helpful discussions and comments.
LH is supported by the 
SFB 676 ``Particles, Strings, and the Early Universe'' of the German Research Foundation (DFG).

\section{Solution to a set of difference equations}
\label{section-soldefeqn}

For two functions $\vec{F} : \mathbb{C} \to \mathrm{Mat}(N,\mathbb{C})$ with components $F_{nm}$ and  $\vec{g}:\mathbb{R}\rightarrow\mathbb{C}^N$ with components $g_n$ let us formally denote by
\be\left(\vec{F}\star \vec{g}\right)(z):=\int_{-\infty}^\infty \vec{F}(z-t)\cdot \vec{g}(t)\ dt
\ee
the convolution product, where the components of the integrand are
$[\vec{F}(z-t)\cdot \vec{g}(t)]_n = \sum_{m=1}^N F_{nm}(z-t) g_m(t)$.
In this section we will prove the following important proposition which will allow us to relate finite difference equations and integral equations.  Recall from Notations~\ref{intro-notation} the definition of the subset
$\mathrm{Mat}_{<2}(N) \subset \mathrm{Mat}(N,\mathbb{R})$, and that we fixed constants $N \in \mathbb{Z}_{>0}$ and $s \in \mathbb{R}_{>0}$. Recall also the definition of $\Phi_{\vec C}$ from \eqref{phiC-def-intro}.

\begin{proposition}\label{FuncRelToNLIE}
Let $\vec{C}\in\mathrm{Mat}_{<2}(N)$,
 $\vec{f} :\mathbb{R}\rightarrow\mathbb{C}^N$ and $\vec{g} \in BC(\mathbb{R},\mathbb{C})^N$.
Consider the following two statements:
\begin{enumerate}
\item \label{ffg-funrel} $\vec{f}$ is real analytic and can be analytically continued to a function $\vec{f}\in \mathcal{BA}(\mathbb{S}_s)^N$ satisfying the functional equation
\be\label{ffuncrel}
\vec{f}(x+is)+\vec{f}(x-is) -\vec{C}\cdot\vec{f}(x) = \vec{g}(x) 
\qquad \text{for all}~~ x\in\mathbb{R}\ .
\ee
\item \label{ffg-inteq} $\vec{f}$ and $\vec{g}$ are related via the convolution
\be\label{fconv}
\vec{f}(x) = \left(\Phi_{\vec C}\star \vec{g}\right)(x)\hspace{1cm} \qquad \text{for all}~~ x\in\mathbb{R} \ .
\ee
\end{enumerate}
We have that \ref{ffg-funrel} implies \ref{ffg-inteq}. If the components of $\vec g$ are in addition H\"older continuous, then \ref{ffg-inteq} implies \ref{ffg-funrel}.
\end{proposition}

Such results have been widely used in the physics literature, but the only rigorous proof we are aware of is in \cite[Lem.\,2]{TraceyWidom} for the special case $N=1$ and
 $\vec{C}=0$, and under
	a decay condition on $\vec{f}(x)$.
  Hence, we will give a complete proof here.

The basic reasoning of our proof is the same as in the physics literature, as outlined in the introduction. We take care to prove all the required analytical properties, which to our knowledge has not been done before in this generality. We also note that the observation that Proposition~\ref{FuncRelToNLIE} applies to all $\vec{C} \in \mathrm{Mat}_{<2}(N)$
(rather than just $\vec{C}=0$ and
adjacency matrices of certain graphs) seems to be new.
The proof relies on a number of ingredients developed in Sections~\ref{GreensChapter}--\ref{section:phiconvolution}. The proof 
of Proposition~\ref{FuncRelToNLIE} itself is given in Section~\ref{GreenSoln}.

\subsection{The Green's function $\phi_d(z)$}
\label{GreensChapter}

In this subsection we introduce a family of meromorphic functions $\phi_d(z)$, parametrized by a real number
\be\label{d_in_(-2,2)}
	d\in(-2,2) \ .
\ee	 
In this section we adopt the convention that whenever the parameter $d$ appears, it is understood that it is chosen from the above range.

The function $\phi_d$ will be central to our problem since it plays, in analogy with the theory of differential equations, the role of a Green's functions for the difference operator
\be\label{diffop}
D[f](x) \,:=\, f(x+is) + f(x-is) - d\cdot f(x) \ .
\ee

We start by defining the function $\phi_d$ on the real axis in terms of a Fourier integral representation.

\begin{definition}
The function $\phi_d:\mathbb{R}\rightarrow\mathbb{R}$ is defined by
\be\label{deltastandardrep}
\phi_d(x):=\frac{1}{2\pi}\int_{-\infty}^{\infty} \frac{e^{ikx}}{2\cosh(sk)-d} \, dk .
\ee
\end{definition}
Note that $\phi_d(x)$ is real and even, since it is the Fourier transform of a real and even function. Moreover, $(2\cosh(sk)-d)^{-1}$ is in the Schwartz space of rapidly decaying functions; thus also $\phi_d(x)$ is a Schwartz function, and Fourier inversion holds.

\begin{example} \label{phi_0} Consider the case $d=0$. The Fourier integral can then be computed explicitly, with the result
\be
\phi_0(z) = \frac{1}{4s\cosh\left(\frac{\pi}{2s}z\right)}\ .
\ee
This function is called the \textsl{universal kernel} or \textsl{standard kernel} in the physics literature.
It has a meromorphic continuation to the whole complex plane, with poles of first order in $z=is(1+2\mathbb{Z})$.
\end{example}

Explicit expressions for $\phi_d$ in terms of hyperbolic functions for other specific values of $d$ can be found e.g.\ in~\cite[App.\,D]{Dorey:2007zx} and \cite[Eqn.\,4.22]{BLZ4}. A general expression is given in Remark~\ref{rem:should-have-found-this-earlier} below.
Here we will proceed by deducing the analytic structure of $\phi_d(z)$ from its definition in \eqref{deltastandardrep}.

Recall that smoothness of a function is related to the rate of decay of its Fourier transform. If the decay is exponential, analytic continuation is possible; the faster the exponential decay, the further one can analytically continue: 

\begin{lemma} \label{fourieranalytic} Let $f(x)$ be a function on $\mathbb{R}$ whose Fourier transform $\hat{f}(k)$ exists. Suppose there exist constants $A,a>0$ such that $|\hat{f}(k)|\leq A\exp(-a|k|)$ for all $k\in\mathbb{R}$. Then $f(x)$ has an analytic continuation to the strip $\mathbb{S}_a$.
\end{lemma}

	For a proof see for example \cite[Ch.\,4,\,Thm.\,3.1]{SteinShakarchi}.
In particular, $\phi_d(x)$ can be analytically continued to the strip $\mathbb{S}_s$. 
In fact, it can be continued to a meromorphic function with poles of order $\le 1$ in $is\mathbb{Z}$ (Lemma~\ref{phianalyticstructure} below). To get there we need some preparation. We start with:

\begin{lemma} \label{phianalyticcontinuation} $\phi_d(z)$ has an analytic continuation to the complex plane with two cuts, 
$\mathbb{C}\setminus (i\mathbb{R}_{\geq s}\cup i\mathbb{R}_{\leq -s})$.
\end{lemma}

\begin{proof}
Take any $\theta\in(-\frac{\pi}{2},\frac{\pi}{2})$ and consider the function
\be\label{phitilde}
\tilde{\phi}_d(z)=e^{-i\theta}\frac{1}{2\pi}\int_{-\infty}^{\infty} e^{ike^{-i\theta}z}\left(2\cosh\left(ske^{-i\theta}\right)-d\right)^{-1}dk\ .
\ee
By Lemma~\ref{fourieranalytic}, this integral is analytic in $z\in e^{i\theta}\mathbb{S}_{s\cdot\cos\theta}$, a strip in the $z$-plane tilted by the angle $\theta$. We claim that $\tilde{\phi}_d(z)$ and $\phi_d(z)$ coincide in the intersection $e^{i\theta}\mathbb{S}_{s\cdot\cos\theta} \cap \mathbb{S}_s$ of their respective analytic domains. 
This can be checked via contour deformation. We will first show that for $z \in e^{i\theta}\mathbb{S}_{s\cdot\cos\theta} \cap \mathbb{S}_s$ we have
\be\label{phirotatecontour}
\phi_d(z) =\frac{1}{2\pi}\int_{e^{-i\theta}\mathbb{R}} e^{ikz}(2\cosh(sk)-d)^{-1}dk \ .
\ee
Since for $|d|<2$, the integrand $k \mapsto e^{ikz}(2\cosh(sk)-d)^{-1}$ has no poles away from the imaginary axis, rotating the contour by $-\theta$ does not pick up any residues. It remains to verify that there are no contributions from infinity.
We express $z=x+iy$ and $k=u+iv$ in real coordinates, in terms of which the absolute value of the integrand can be written as
\be\label{estimintegrand}
\left|\frac{e^{ikz}}{2\cosh(sk)-d}\right|=\left|\frac{e^{-(uy+vx)}}{e^{su}e^{isv}+e^{-su}e^{-isv}-d}\right| \ .
\ee
On the circular contour components one can parametrise $v=-u\tan(\tau)$, with $\tau$ running from 0 to $\theta$. When $u\rightarrow\pm\infty$, the right hand side of \eqref{estimintegrand} approaches
\be
\left|e^{-u(y\pm s)-vx}\right|=\left|e^{-u(y\pm s-\tan(\tau) x)}\right| \ .
\ee
Thus, if the inequalities
\begin{align}
y&>\tan(\tau) x -s  \nonumber \\
y&<\tan(\tau)  x +s 
\end{align}
are satisfied for all $\tau$ between 0 and $\theta$, then the two circular integrals do indeed vanish when the radius is taken to infinity. But these inequalities just describe the strips $e^{i\tau}\mathbb{S}_{s\cdot\cos\tau}$ in the $z$-plane, and their intersection for all $\tau \in [0,\theta]$ is precisely $\mathbb{S}_s \cap e^{i\theta}\mathbb{S}_{s\cdot\cos\theta}$. 

Now, substituting $k'=e^{i\theta}k$ in \eqref{phirotatecontour} shows that $\phi_d(z) = \tilde \phi_d(z)$ on the intersection of their domains, and hence $\tilde{\phi}_d(z)$ is the analytic continuation of $\phi_d(z)$ to the strip $e^{i\theta}\mathbb{S}_{s\cdot\cos\theta}$.
Consequently, $\phi_d(z)$ has an analytic continuation to the union of all of these strips,
\be
\bigcup_{\theta\in(-\frac{\pi}{2},\frac{\pi}{2})}e^{i\theta}\mathbb{S}_{s\cdot\cos\theta} = \mathbb{C}\setminus (i\mathbb{R}_{\geq s}\cup i\mathbb{R}_{\leq -s})\ .
\ee
\end{proof}

In order to understand the behaviour of $\phi_d(z)$ on the whole imaginary axis, it is natural to start with the case $d=0$ whose analytic structure we know explicitly (Example~\ref{phi_0}). When comparing $\phi_d(z)$ to $\phi_0(z)$ we will need to control the derivative $\frac{\partial}{\partial d}\phi_d(z)$.

\begin{lemma} \label{dderivative} For all $z \in \mathbb{C}\setminus (i\mathbb{R}_{\geq s}\cup i\mathbb{R}_{\leq -s})$,
the partial derivative $\frac{\partial}{\partial d}\phi_d(z)$ exists and is 
an analytic function on $\mathbb{C}\setminus (i\mathbb{R}_{\geq 2s}\cup i\mathbb{R}_{\leq -2s})$. For $z\in e^{i\theta}\mathbb{S}_{2s\cdot \cos\theta}$ it has the integral representation
\be\label{eqn:dderivative}
\frac{\partial}{\partial d}\phi_d(z) = e^{-i\theta}\frac{1}{2\pi} \int_{-\infty}^\infty e^{ike^{-i\theta}z} \left(2\cosh(ske^{-i\theta}) -d\right)^{-2} dk \ .
\ee
\end{lemma}

\begin{proof}
For any $z\in e^{i\theta}\mathbb{S}_{s\cdot \cos\theta}$, consider $\phi_d(z)$  given by the integral representation \eqref{phitilde}.
Write $ke^{-i\theta}= u+iv$ with $u,v\in\mathbb{R}$. Assuming $u\geq 0$, one then estimates
\begin{align}
&\left|2\cosh\left(ske^{-i\theta}\right)-d\right|
~=~
\left|e^{su}e^{isv}+e^{-su}e^{-isv}-d\right|\nonumber\\
&\geq~
\left|e^{su}e^{isv}\right|-\left|e^{-su}e^{-isv}\right|-|d|
~\geq~ e^{su}-1-\left|d\right|
~\geq~ e^{s\cos\theta \, |k|}-3 \ .
\end{align}
The same overall estimate applies for $u<0$ as well. For $|k|$ large enough, the last expression on the right hand side becomes bigger than $\frac{1}{2}e^{s\cos\theta \,|k|}$ and we can estimate
\be\label{cosh-theta-estimate}
	\left|2\cosh\left(ske^{-i\theta}\right)-d\right|^{-1}
	~\le~
	2 \, e^{-s\cos\theta \,|k|} 
	\qquad  \text{(for $|k|$ large enough)} \ .
\ee
Next, writing $z = e^{i \theta} (x + i y)$ with 
	$x\in\mathbb{R}$ and $y \le s\cos\theta$
we obtain, for all $k \in \mathbb{R}$,
\be
	\left|e^{ike^{-i\theta}z}\right|
	~=~ e^{-ky} 
	~\le~ e^{s\cos\theta \,|k|} \ .
\ee
For $|k|$ large enough and $z\in e^{i\theta}\mathbb{S}_{s\cdot \cos\theta}$ we can now estimate
\begin{align}\left|\frac{\partial}{\partial d}e^{ike^{-i\theta}z} \left(2\cosh(ske^{-i\theta}) -d\right)^{-1}\right| &= 
	\left|e^{ike^{-i\theta}z}\right| \cdot
\left| 2\cosh(ske^{-i\theta}) -d\right|^{-2} \nonumber\\
&\leq 4 \,  e^{-s\cos\theta \,|k|} \ .
\end{align}
One can choose 
$k_0>0$ large enough, such that this estimate applies for all $d \in (-2,2)$
and $|k| \ge k_0$.

Let $\varepsilon>0$, and define $D:=2-\varepsilon$. Since for any $k_0>0$, the integrand of \eqref{eqn:dderivative} is continuous (and finite) as a function of $(k,d)$ in the compact region $[-k_0,k_0]\times [-D,D]$, it is in particular bounded. One can therefore find a constant $A>0$ such that the integrand of  \eqref{eqn:dderivative}  is bounded by $A \,e^{-s\cos\theta \,|k|}$ for all $k \in \mathbb{R}$ and $d \in [-D,D]$.

The integrand is thus majorised by the integrable function $A \,e^{-s\cos\theta \,|k|}$ for all $d \in [-D,D]$ and therefore  integration and $d$-derivative can be swapped, establishing \eqref{eqn:dderivative} for all $d \in [-D,D]$ and $z\in e^{i\theta}\mathbb{S}_{s\cdot \cos\theta}$.
Since $\varepsilon>0$ was arbitrary, this extends to all $d\in(-2,2)$,
proving the first part of the claim. 
Moreover, according to Lemma \ref{fourieranalytic}, the integral on the right-hand-side of \eqref{eqn:dderivative} is actually analytic for $z\in e^{i\theta}\mathbb{S}_{2s\cdot \cos\theta}$. By uniqueness of the analytic continuation it must also coincide with $\frac{\partial}{\partial d}\phi_d(z)$ on this larger domain.
\end{proof}

One consequence of the integral representation of the
$d$-derivative is the following functional relation for $\phi_d$.

\begin{lemma} \label{phifuneq} For all $x\in\mathbb{R}\setminus \lbrace 0 \rbrace$ we have
\be\label{eqn:phifuneq}\phi_d(x+is)+\phi_d(x-is)-d\phi_d(x)=0 \ .
\ee
\end{lemma}

\begin{proof}
	Fix $x \in \mathbb{R}$, $x \neq 0$, and 
	write
\be
\mathcal{L}_x(d):=\phi_d(x+is)+\phi_d(x-is)-d\phi_d(x) \ .
\ee
We will show that $\mathcal{L}_x(d)$ solves the initial value problem  \be \label{dgl36}
 \mathcal{L}_x(0)=0 
 \quad , \quad \frac{\partial}{\partial d}\mathcal{L}_x(d) =0 
 \qquad \text{for all $d\in (-2,2)$} \ .
 \ee 

The initial condition in \eqref{dgl36} is
	straightforward 
to check by recalling from Example~\ref{phi_0} that $\phi_0(z)=\left(4s\cosh\left(\frac{\pi}{2s}z\right)\right)^{-1}$. In order to prove the differential equation in \eqref{dgl36}, we use the integral representation \eqref{eqn:dderivative} for
$\theta=0$ and $z\in\mathbb{S}_{2s}$:
\begin{align}
\frac{\partial}{\partial d}\mathcal{L}_x(d) &= \frac{\partial}{\partial d}\phi_d(x+is)+\frac{\partial}{\partial d}\phi_d(x-is)-\phi_d(x)-d\frac{\partial}{\partial d}\phi_d(x) \nonumber\\
&=\frac{1}{2\pi} \int_{-\infty}^\infty e^{ikx} \frac{F(k)}{\left(2\cosh(sk) -d\right)^2}dk,
\end{align}
where
\be
F(k)=e^{-ks}+e^{ks}-(2\cosh(sk)-d)-d =0 \ .
\ee
Hence, \eqref{dgl36} follows.
\end{proof}

By Lemma \ref{phianalyticcontinuation}, the function
	$x \mapsto \mathcal{L}_x(d)$
has an analytic continuation
to all $z\in\mathbb{C}\setminus i\mathbb{R}$. The functional relation \eqref{eqn:phifuneq} thus extends to this domain:
\be\label{fun-rel-phid-for-z}
\phi_d(z+is)+\phi_d(z-is)-d\phi_d(z)=0 
\qquad 
\text{for all} \quad
z\in\mathbb{C}\setminus i\mathbb{R} \ .
\ee
Using this, we now show:

\begin{lemma} \label{phianalyticstructure} 
$\phi_d(z)$ has a meromorphic continuation to the whole complex plane which satisfies:
\begin{enumerate}[label=\roman*)]
\item \label{Poles1stOrder} The poles are all of first order and form a subset of $is\mathbb{Z}\setminus\lbrace 0\rbrace$. 
\item \label{extended_funrel}
For $z \in \mathbb{C} \setminus is\mathbb{Z}$ we have $\phi_d(z+is)+\phi_d(z-is)-d\phi_d(z)=0$.
\item \label{recurs}
For $z \in \mathbb{C} \setminus is\mathbb{Z}$ and $n \in \Zb$ we have
\be\label{eqn:phirecursive}
\phi_d(z+ ins) = \frac{\sin(n\gamma)}{\sin(\gamma)}\phi_d(z+is)-\frac{\sin((n-1)\gamma)}{\sin(\gamma)}\phi_d(z) \ ,
\ee
where $\gamma \in \mathbb{R}$ satisfies
$d=2\cos(\gamma)$.
\end{enumerate}
\end{lemma}

\begin{proof}~\\
$\bullet$ {\em Relation \eqref{eqn:phirecursive} holds on
$\mathbb{C}\setminus i\mathbb{R}$}: Writing $p_n(z):=\phi_d(z+ins)$, we can rewrite \eqref{fun-rel-phid-for-z} as $p_{n+1}+p_{n-1}-dp_n=0$. 
 It is straight-forward to check that this recursion relation is solved by \be
 p_n=\frac{\sin(n\gamma)}{\sin(\gamma)}p_1-\frac{\sin((n-1)\gamma)}{\sin(\gamma)}p_0\ .
 \ee
$\bullet$ {\em $\phi_d$ has an analytic continuation to $\mathbb{C}$ minus the points $is\mathbb{Z}\setminus\lbrace 0\rbrace$}: 
The right hand side of \eqref{eqn:phirecursive} is actually analytic in $\{ z \in \mathbb{C} | -s < \mathrm{Im}(z) < 0 \}$. Hence the same holds for the left hand side, that is, $\phi_d$ is analytic on the shifted strips $\{ z \in \mathbb{C} | sn < \mathrm{Im}(z) < (n+1)s \}$ for all $n \in \mathbb{Z}$. Combining this with Lemma~\ref{phianalyticcontinuation}, which states that $\phi_d(z)$ is analytic on $\mathbb{C}\setminus (i\mathbb{R}_{\geq s}\cup i\mathbb{R}_{\leq -s})$, we obtain the claim.

In particular, \eqref{fun-rel-phid-for-z} and \eqref{eqn:phirecursive} hold on $\mathbb{C} \setminus is\mathbb{Z}$, showing parts \ref{extended_funrel} and \ref{recurs} of the lemma.

\medskip

\noindent
$\bullet$ {\em $\phi_d$ is Lipschitz continuous in $d$ on $[-D,D]$ for $D<2$:}
Consider again the case $d=0$, given by $\phi_0(z)=\left(4s\cosh\left(\frac{\pi}{2s}z\right)\right)^{-1}$: this is a meromorphic function with poles of first order located at $is(1+2\mathbb{Z})$. 
We are now going to show that in the strip $\mathbb{S}_{2s}$ the pole structure of $\phi_d(z)$ for general $d$ coincides with the pole structure of $\phi_0(z)$. 

Recall the derivative $\frac{\partial}{\partial d}\phi_d(z)$ defined inside $\mathbb{S}_{2s}$ by the integral representation \eqref{eqn:dderivative}. We 
write $y=\mathrm{Im}(z)$, where $y\in(-2s,2s)$, and obtain 
\be
\left|\frac{\partial}{\partial d}\phi_d(z)\right|\leq \frac{1}{2\pi} \int_{-\infty}^\infty\frac{e^{-ky}}{\left(2\cosh(sk)-d\right)^2}dk \ .
\ee 
	Let $\varepsilon>0$ be arbitrary so that $D:=2-\varepsilon>0$ and $Y:=2s-\varepsilon>0$. For $d\leq D$ and $|y|\leq Y$, the integrand can then be further estimated as follows:
\be\frac{e^{-ky}}{\left(2\cosh(sk)-d\right)^2}\leq \frac{2\cosh(yk)}{\left(2\cosh(sk)-d\right)^2}\leq \frac{2\cosh(Yk)}{\left(2\cosh(sk)-D\right)^2}\ee But \be 
B:=\frac{1}{2\pi}\int_{-\infty}^\infty \frac{2\cosh(Yk)}{\left(2\cosh(sk)-D\right)^2}dk\ee still converges (as $|Y|<2s$), and hence \be\left|\frac{\partial}{\partial d}\phi_d(z)\right|\leq B \ .\ee Put differently, $B$ is a Lipschitz constant for $\phi_d(z)$ understood as a function of $d$ on the interval $[-D,D]$. The Lipschitz condition reads \be\left|\phi_d(z)-\phi_0(z)\right|\leq B\cdot d \ .
\ee

\medskip

\noindent
$\bullet$ {\em The pole order of $\phi_d$ at $is\Zb \setminus \{0\}$ is 0 or 1}:
By Lipschitz-continuity in $d$, we know that $|\phi_d(z)-\phi_0(z) |\leq B\cdot d$ for all $z \in \mathbb{S}_{2s} \setminus \{ \pm is\}$. Since $\phi_0$ has first order poles at $\pm is$, it follows that so does $\phi_d$. 
Since \eqref{eqn:phirecursive} holds on $\mathbb{C} \setminus is\mathbb{Z}$, the pole structure is as claimed.
This finally proves part \ref{Poles1stOrder} of the lemma.

We remark that the Lipschitz-continuity in $d$ does not extend beyond the strip $\mathbb{S}_{2s}$. Indeed, 
$\phi_0$ has no pole at $\pm 2is$, whereas for $d\neq 0$,  \eqref{eqn:phifuneq} forces $\phi_d(z)$ to have a pole there. 
\end{proof}

Next we turn to the growth properties of $\phi_d$. We need the following result from complex analysis (see e.g.\ \cite[Ch.\,4,\,Thm.\,3.4]{SteinShakarchi} for a proof).

\begin{theorem}[Phragm\'en-Lindel\"of]\label{phrag-lind}
Let $f$ be a holomorphic function in the wedge $W = \{ z \in \mathbb{C} \,|\, -\tfrac\pi{2\beta} < \mathrm{arg}(z) < \tfrac\pi{2\beta} \}$, $\beta>\tfrac12$, which is continuous on the closure of $W$. Suppose that $|f(z)| \le 1$ on the boundary of $W$ and that there are $A,B>0$ and $0<\alpha<\beta$ such that $|f(z)| \le A e^{B |z|^\alpha}$ for all $z \in W$. Then $|f(z)| \le 1$ on $W$.
\end{theorem}

With the help of this theorem we can establish the following boundedness properties.

\begin{lemma} \label{phibounded} Let $\Theta\in[0,\frac{\pi}{2})$. Then for all $m,n\in\mathbb{Z}_{\ge 0}$ 
the function $z^m\frac{d^n}{dz^n}\phi_d(z)$ is bounded in the wedges $|\arg(z)|\leq\Theta$ and $|\arg(z)-\pi|\leq\Theta$.
\end{lemma}

\begin{proof}
Consider the lines $z=e^{\pm i\Theta}\mathbb{R}$, which constitute the boundary of the wedges we are interested in. The integral representation \eqref{phitilde} of $\phi_d(z)$, when restricted to these lines, yields
 functions in the Schwartz space, because $\left(2\cosh\left(ske^{\mp i\Theta}\right)-d\right)^{-1}$ are functions in the Schwartz space. By definition of the Schwartz space, the function $z^m\frac{d^n}{dz^n}\phi_d(z)$ is bounded on these two lines as well. Moreover, it is analytic in the interior of the wedges. We will now show that the growth of $z^m\frac{d^n}{dz^n}\phi_d(z)$ is less than exponential in the interior of the wedges. The statement of the lemma then follows from  Theorem~\ref{phrag-lind}.

It suffices to show that $\frac{d^n}{dz^n}\phi_d(z)$ is bounded in the wedges $|\arg(z)|<\Theta$ and $|\arg(z)-\pi|<\Theta$. Let $x\in \mathbb{R}$, $\theta\in(-\Theta,\Theta)$, and use the integral representation \eqref{phitilde} to obtain the $x$-independent estimate
\be\label{108u3}
\left|\frac{d^n}{dz^n}\phi_d(xe^{i\theta})\right| 
\leq \frac{1}{2\pi} \int_{-\infty}^{\infty} \left|\frac{k^n}{2\cosh\left(ske^{-i\theta}\right)-d}\right|dk \ .
\ee
Next we estimate the integrand by a $\theta$-independent integrable function. 
Recall from \eqref{cosh-theta-estimate} that for $|k|$ large enough we can estimate
\be
	\left|2\cosh\left(ske^{-i\theta}\right)-d\right|^{-1}
	~\le~
	2 \, e^{-s\cos\theta \,|k|} 
	~\le~
	2 \, e^{-s\cos\Theta \,|k|} \ .
\ee
But for $|k|$ large enough, we also have $|k^n|\leq e^{\frac{s}{2}\cos\Theta\, |k|}$. In other words, there exists a $k_0>0$, independent of $\theta$, such that for all $|k|\geq k_0$,
\be\label{vs98h}
\left|\frac{k^n}{2\cosh\left(ske^{-i\theta}\right)-d}\right|\leq 2e^{-\frac{s}{2}\cos\Theta\, |k|} \ .
\ee
Meanwhile, the function $w \mapsto \left|2\cosh(sw)-d\right|$ is continuous in 
	$\lbrace w\in\mathbb{C}|\, |\arg(w)|\leq\Theta, |\mathrm{Re}(w)|\leq k_0 \rbrace $ and in the corresponding wedge with $|\arg(-w)|\leq\Theta$,
and has no zeros in this
bow tie shaped
compact subset of the complex plane. Hence, it is bounded from below by a strictly positive number $M>0$.
Consequently, for all $|k|<k_0$,
\be\label{brt65}
\left|\frac{k^n}{2\cosh\left(ske^{-i\theta}\right)-d}\right|\leq \frac{k_0^n}{M} \ .
\ee
Plugging \eqref{vs98h} and \eqref{brt65} into \eqref{108u3} yields the bound \be\left|\frac{d^n}{dz^n}\phi_d(z)\right|\leq \frac{1}{2\pi}\left(\int_{-k_0}^{k_0}\frac{k_0^n}{M}dk + 2 \int_{k_0}^{\infty} 2e^{-\frac{s}{2}\cos\Theta\, k} dk \right) 
\ ,
\ee 
valid for all $z$ in the two wedges defined by $\Theta$, and 
where the right hand side is finite and independent of $\theta$.
\end{proof}

\begin{corollary} \label{phiL1} For $y\in\mathbb{R}$, let $\partial^n\phi_d^{[y]}(x):=\frac{d^n}{dz^n}\phi_d(x+iy)$ be the restrictions of $\frac{d^n}{dz^n}\phi_d(z)$ to horizontal lines. If $y\notin s\mathbb{Z}\setminus\lbrace 0 \rbrace$, then 
$(x+iy)^m\partial^n\phi_d^{[y]}(x)\in L_1(\mathbb{R})$ 
for all $n,m\in\mathbb{Z}_{\geq 0}$.
\end{corollary}

\begin{corollary} \label{boundaryintegrals} For any $a,b\in\mathbb{R}$ and for all $n,m\in\mathbb{Z}_{\geq 0}$, \be\lim_{x\rightarrow \pm\infty}\int_a^b(x+it)^m\frac{d^n}{dz^n}\phi_d(x+it)\, dt =0 \ .\ee
\end{corollary}

To understand at which points of $is\mathbb{Z} \setminus \{0\}$ the function $\phi_d$ has a first order pole and at which points the singularity can be lifted, we compute the residues.

\begin{lemma} \label{phiresidue} 
For $n \in \mathbb{Z}$, the residue of $\phi_d(z)$ in $z=isn$  is given by 
\be
\mathrm{Res}_{isn}(\phi_d) = \frac{1}{2\pi i} \frac{\sin(n\gamma)}{\sin(\gamma)}\ ,
\ee 
where $\gamma \in \mathbb{R}$ satisfies $d=2\cos(\gamma)$.
\end{lemma}

\begin{proof}
We start by computing the residue at $is$:
\begin{align}
2 \pi i \, \mathrm{Res}_{is}(\phi_d) &
\overset{(a)}=\int_{\mathbb{R}+\frac{1}{2}is}\phi_d(z)\,dz - \int_{\mathbb{R}+\frac{3}{2}is}\phi_d(z)\,dz \nonumber\\
&= \int_{-\infty}^{\infty}\phi_d\left(x+\tfrac{is}{2}\right)dx -\int_{-\infty}^{\infty}\phi_d\left(x+\tfrac{3is}{2}\right)dx \nonumber\\
&\overset{(b)}=\int_{-\infty}^{\infty}\phi_d\left(x+\tfrac{is}{2}\right)dx +\int_{-\infty}^{\infty}\phi_d\left(x-\tfrac{is}{2}\right)dx-d\int_{-\infty}^{\infty}\phi_d\left(x+\tfrac{is}{2}\right)dx \nonumber\\
&\overset{(c)}=(2-d)\int_{-\infty}^{\infty}\phi_d(x)dx = (2-d)\frac{1}{2\cosh(sk)-d}\Big|_{k=0}
= 1 \ .
\end{align}
Here, all integrals exist by Corollary~\ref{phiL1}. In step (a), the circular contour is deformed to two horizontal infinite lines, making use of Corollary~\ref{boundaryintegrals} to ensure that no contribution is picked up when pushing the vertical parts of the contour to infinity. Step (b) is the functional relation in
 Lemma~\ref{phianalyticstructure}\,\ref{extended_funrel}. In step (c) all contours are moved to the real axis, using that $\phi_d$ is analytic in $\mathbb{S}_{s}$ (Lemma~\ref{phianalyticstructure}) and that by Corollary~\ref{boundaryintegrals} there are no contributions from infinity.

But from \eqref{eqn:phirecursive} we see that \be\mathrm{Res}_{isn}\phi_d = \frac{\sin(n\gamma)}{\sin(\gamma)}\mathrm{Res}_{is}\phi_d-\frac{\sin((n-1)\gamma)}{\sin(\gamma)}\mathrm{Res}_{0}\phi_d\ .\ee
Since $\mathrm{Res}_{0}\phi_d=0$ we obtain the statement of the lemma.
\end{proof}

We are now in a position to justify the notion that $\phi_d(z)$ is a Green's function for the difference operator \eqref{diffop}:

\begin{lemma} \label{Greensproperty} $\phi_d(x)$ gives rise to a representation of the Dirac $\delta$-distribution on $BC(\mathbb{R},\mathbb{C})$ via
\be
\lim_{y\nearrow  s}\left(\phi_d(x+iy) + \phi_d(x-iy) - d\phi_d(x) \right)= \delta(x)\ .
\ee
\end{lemma}

Before we turn to the proof, we note that for $|y|<s$, 
\begin{align}
\phi_d(x+iy) + \phi_d(x-iy) - d\phi_d(x) &= \frac{1}{2\pi}\int_{-\infty}^{\infty} \frac{e^{ikx-yk}+e^{ikx+yk}-de^{ikx}}{2\cosh(sk)-d}dk \nonumber\\
&=\frac{1}{2\pi}\int_{-\infty}^{\infty}e^{ikx} \frac{2\cosh(yk)-d}{2\cosh(sk)-d}dk\ .
\end{align}
In the limit $y\nearrow s$, the integrand on the right hand side approaches $e^{ikx}$ pointwise. The usual exchange-of-integration-order argument proves that one obtains a Dirac $\delta$-distribution on $L_1$-functions whose Fourier-transformation is also $L_1$. To show that we obtain a $\delta$-distribution on $BC(\mathbb{R},\mathbb{C})$, we follow a different route.

\begin{proof}[Proof of Lemma~\ref{Greensproperty}]
Since $\phi_d(z)$ has simple poles at $z=\pm is$ of residue $\pm \frac{1}{2\pi i}$ (see Lemma \ref{phiresidue}), we can write 
\be\phi_d(z)= \pm\frac{1}{2\pi i}\frac{1}{(z\mp is)}+r_\pm(z)\ ,
\ee 
where $r_\pm(z)$ is now analytic at
$z=\pm is$. In particular, by Lemma~\ref{phibounded}
$r_+(z)$ is bounded in the upper half of $\mathbb{S}_s$ and $r_-(z)$ is bounded in the lower half of $\mathbb{S}_s$. Then
\begin{align}
\delta_y(x) &:= \phi_d(x+iy)+\phi_d(x-iy)-d\phi_d(x) \nonumber\\
&= \frac{1}{2\pi i}\left(\frac{1}{(x+iy- is)}-\frac{1}{(x-iy+ is)}\right)+r_+(x+iy)+r_-(x-iy)-d\phi_d(x) \nonumber\\
&= \tilde{\delta}_y(x) +u_y(x)\ ,
\end{align}
where \be\tilde{\delta}_y(x):=\frac{1}{\pi}\frac{s-y}{x^2 + (s-y)^2}\ee and $u_y(x):=r_+(x+iy)+r_-(x-iy)-d\phi_d(x)$.

Now suppose $f\in BC(\mathbb{R},\mathbb{C})$. We have to show that \be
\lim_{y\nearrow s} \int_{-\infty}^\infty \delta_y(x) f(x) dx = f(0).
\ee 
In order to do that, let $\varepsilon >0$ be arbitrary and split the integral, 
\be\int_{-\infty}^\infty \delta_y(x) f(x) dx = I_1(y)+I_2(y)+I_3(y)\ ,
\ee 
where 
\be
I_1(y)=\int_{\mathbb{R}\setminus (-\varepsilon,\varepsilon)} \delta_y(x) f(x) dx\ , 
\ee 
\be
I_2(y)= \int_{-\varepsilon}^\varepsilon \tilde{\delta}_y(x) f(x) \; dx, \hspace{1cm} I_3(y)= \int_{-\varepsilon}^\varepsilon u_y(x) f(x) dx\ .
\ee The \textsl{Lorentz functions} $\tilde{\delta}_y(x)$ are a well-known representation of the Dirac $\delta$-distribution on $BC(\mathbb{R},\mathbb{C})$ as $y\rightarrow s$, so independently of $\varepsilon$ we have \be\lim_{y\nearrow s} I_2(y)=f(0)\ .\ee  The functions $u_y(x)$ are uniformly bounded for $y\in[0,s]$, say by $C$. Hence, \be|I_3(y)|\leq 2\varepsilon C \left\|f\right\|_\infty \ .\ee Finally, consider $\delta_y(x)$ on $\mathbb{R}\setminus(-\varepsilon,\varepsilon)$. Due to the functional equation \eqref{eqn:phifuneq}, it converges pointwisely to zero on this domain as $y\rightarrow s$. But Lemma \ref{UniformConvergenceOfAnalyticFuns} in connection with Lemma \ref{phibounded} even ensures uniform convergence, and this is still true for the function $x^2 \delta_y(x)$. By means of the variable transformation $t=1/x$ we can recast $I_1(y)$ as an integral over a finite interval:
\be
I_1(y)=\int_{-\tfrac{1}{\varepsilon}}^{\tfrac{1}{\varepsilon}} \frac{1}{t^2}\delta_y(\tfrac{1}{t}) f(\tfrac{1}{t}) dt\ , 
\ee 
By what we just said, the integrand converges uniformly to zero on $[-\tfrac{1}{\varepsilon},\tfrac{1}{\varepsilon}]$.
Thus, integral and limit can be swapped, which results in
\be\lim_{y\nearrow s} I_1(y)=0\ .\ee 
We conclude that \be\left|f(0) -\lim_{y\nearrow s}\int_{-\infty}^\infty \delta_y(x) f(x) dx \right|\leq 2\varepsilon C\left\|f\right\|_\infty\ .\ee As $\varepsilon>0$ was arbitrary, the statement follows.
\end{proof}

\begin{remark} \label{rem:should-have-found-this-earlier}
After completion of this paper we noticed that one can actually give a simple explicit expression for $\phi_d(z)$ for arbitrary $d\in(-2,2)$:
\be\label{explicit_phi_d}
\phi_d(z)~=~\frac{1}{2s\sin(\gamma)}\cdot\frac{\sinh\left(\frac{\pi-\gamma}{s}z\right)}{\sinh\left(\frac{\pi}{s}z\right)} \ ,
\ee
where $\gamma\in(0,\pi)$ is defined via $d=2\cos(\gamma)$. This can be seen via a contour deformation argument using the analytical properties of $\phi_d$ established in this section, we will provide the details elsewhere. The explicit integral can also be found in tables,
see e.g.\ \cite[1.9\,(6)]{Erdelyi}.
We were, however, unable to obtain it in a more straight-forward fashion, circumventing the analysis carried out in this section.
\end{remark}

\subsection{The $N$-dimensional Green's function $\Phi_{\vec C}(z)$}
\label{section-NdimGreen}

Now let us investigate an $N$-dimensional version of the Green's function $\phi_d(z)$.
Recall from Notations~\ref{intro-notation} the definition of the subset $\mathrm{Mat}_{<2}(N) \subset \mathrm{Mat}(N,\mathbb{R})$, as well as from \eqref{phiC-def-intro} the definition of 
$\Phi_{\vec C}:\mathbb{R}\rightarrow\mathrm{Mat}(N,\mathbb{R})$ 
for $\vec{C}\in\mathrm{Mat}_{<2}(N)$.

\begin{lemma}\label{Phiproperties} $\Phi_{\vec C}(x)$ has the following properties:
\begin{enumerate}[label=\roman*)]
\item \label{simdiag}$\Phi_{\vec C}(x)$ and $\vec{C}$ are simultaneously diagonalisable for all $x\in\mathbb{R}$.
\item \label{matrixelements} Any matrix element of $\Phi_{\vec C}(x)$ can be written as a linear combination \be [\Phi_{\vec C}]_{nm}(x) = \sum_{j=1}^N \Omega_{nm}^j \phi_{d_j}(x)\ee of one-dimensional Green's functions, with $d_j\in(-2,2)$ given by the eigenvalues of $\vec{C}$, and $\Omega_{nm}^j$ some real coefficients. 
\item \label{delta} $\Phi(x)$ gives rise to a representation of the Dirac $\delta$-distribution on $BC(\mathbb{R},\mathbb{C})^N$:
\be\delta(x)\mathbf{1} = \lim_{y\nearrow  s}\left(\Phi_{\vec C}(x+iy) + \Phi_{\vec C}(x-iy) - \vec{C}\cdot \Phi_{\vec C}(x) \right) \ee 
\end{enumerate}
\end{lemma}

\begin{proof}
First of all, note that the matrix inverse in the definition is well-defined since $\vec{C}$ has spectral radius smaller than 2.

\medskip

\noindent
\ref{simdiag} Let $\vec{D}\in\mathrm{Mat}(N,\mathbb{R})$ be a
diagonal matrix such that $\vec{D} = \vec{T}^{-1}\vec{C}\vec{T}$ for some invertible $\vec{T}\in \mathrm{Mat}(N,\mathbb {R})$. Then
\begin{align}\vec{T}^{-1}\Phi_{\vec C}(x)\vec{T} &~=~ \frac{1}{2\pi}\int_{-\infty}^{\infty} e^{ikx}\;\vec{T}^{-1}(2\cosh(sk)\mathbf{1}-\vec{C})^{-1}\vec{T} \;dk \nonumber\\
&~=~\frac{1}{2\pi}\int_{-\infty}^{\infty} e^{ikx}\left(2\cosh(sk)\mathbf{1}-\vec{D}\right)^{-1}dk 
~=~ \Phi_{\vec D}(x)
\end{align}
is also diagonal.

\medskip

\noindent
\ref{matrixelements} Write $\vec{D}=\mathrm{diag}(d_1,...,d_N)$. Then the matrix elements can be written as \be[\Phi_{\vec C}]_{nm}(x)=\sum_{j=1}^N T_{nj}[T^{-1}]_{jm}[\Phi_D]_{jj}(x) \ee where $\Omega_{nm}^j=T_{nj}[T^{-1}]_{jm}$ are real constants, and $[\Phi_D]_{jj}(x) = \phi_{d_j}(x)$. 
	Since $\vec{C} \in \mathrm{Mat}_{<2}(N)$,
$|d_j|<2$ holds for all $j=1,...,N$.

\medskip

\noindent
\ref{delta} Using \ref{simdiag} and \ref{matrixelements} and applying the Green's function property of $\phi_d(z)$ (Lemma \ref{Greensproperty}), one computes
\begin{align}
\lim_{y\nearrow  s}&\left[\Phi_{\vec C}(x+iy) + \Phi_{\vec C}(x-iy) - \vec{C}\cdot \Phi_{\vec C}(x)\right]_{nm} \nonumber\\ &= \sum_{j=1}^N T_{nj}[T^{-1}]_{jm}\lim_{y\nearrow  s}\left(\phi_{d_j}(x+iy) +\phi_{d_j}(x-iy) -d_j\phi_{d_j}(x)\right) \nonumber\\
&=\sum_{j=1}^N T_{nj}[T^{-1}]_{jm}\delta(x) 
=\delta_{nm}\delta(x) \ .
\end{align}
This completes the proof.
\end{proof}

From the definition of $\phi_d$ in \eqref{deltastandardrep} and from
Lemma~\ref{Phiproperties}\,\ref{matrixelements} we know that all components of $\Phi_{\vec C}(x)$ are Schwartz functions on $\mathbb{R}$ (cf.\ Lemma~\ref{phibounded}). Hence the Fourier transformation of $\Phi_{\vec C}(x)$  reproduces the integrand in \eqref{phiC-def-intro}. In particular, for $k=0$ we obtain the following integral, which we will need later:
\begin{equation}\label{phiC-infinf-integral}
\int_{-\infty}^\infty
\hspace{-.5em}
\Phi_{\vec C}(x) \,dx ~=~
\big( 2 \cdot \mathbf{1}-\vec{C} \big)^{-1} \ .
\end{equation}

\begin{lemma}\label{Phinonneg}
Suppose $\vec{C}$ is non-negative. Then the matrix $\Phi_{\vec C}(x)$ is non-negative for all $x\in\mathbb{R}$.
\end{lemma}

\begin{proof}
The integrand can be expanded into a Neumann series,
\begin{align}
\left(2\cosh(sk)\mathbf{1}-\vec{C}\right)^{-1} &= \left(2\cosh(sk)\right)^{-1}\left(\mathbf{1}-\frac{\vec{C}}{2\cosh(sk)}\right)^{-1}\nonumber\\
&= \sum_{j=0}^\infty \frac{\vec{C}^j}{(2\cosh(sk))^{j+1}}\ ,
\end{align} 
which converges absolutely since for $k\in\mathbb{R}$ all eigenvalues of $\left(2\cosh(sk)\right)^{-1}\vec{C}$ are strictly smaller than 1. Fubini's theorem (with counting measure on $\mathbb{Z}_{\ge0}$ and Lebesgue measure on $\mathbb{R}$) then justifies pulling the sum out of the Fourier integral, and we find that \be\label{phiseries}\Phi_{\vec C}(x)=\sum_{j=0}^\infty \left(\frac{1}{2\pi}\int_{-\infty}^\infty e^{ikx} \left(2\cosh(sk)\right)^{-j-1}dk\right) \vec{C}^j \ .\ee In Appendix~\ref{Fourier_coshm} it is shown that
\be\label{coshm-integral-strictly-positive}
\int_{-\infty}^\infty e^{ikx}\left(2\cosh(sk)\right)^{-j-1}dk = \frac{\pi}{2^{j+1} j!s}\left(\prod_{\substack{l=j-1\\ \mathrm{step -2}}}^1 \left(\frac{x^2}{s^2} + l^2\right)\right)\cdot \begin{cases} \frac{1}{\cosh\left(\frac{\pi}{2s}x\right)} & \mathrm{if} \; j \;\mathrm{even} \\ \frac{x}{s\sinh\left(\frac{\pi}{2s}x\right)} &\mathrm{if} \; j \;\mathrm{odd} \end{cases} \ ,\ee
which is a
strictly positive
function of $x\in\mathbb{R}$. Furthermore, $\vec{C}^j$ is a non-negative matrix. Hence, $\Phi_{\vec C}(x)$ is non-negative for all $x\in\mathbb{R}$.
\end{proof}

\begin{remark}\label{remark-Smatrix-exists}
Suppose $\vec G \in \mathrm{Mat}_{<2}(N)$ is non-negative and irreducible. One of the equivalent ways to characterise irreducibility is that for each $i,j$ there is an $m>0$ such that $[(\vec G)^m]_{ij} \neq 0$. Together with non-negativity of $\vec G$ and strict positivity of \eqref{coshm-integral-strictly-positive}, 
this implies that 
$\Phi_{\vec G}(x)$ has \textsl{strictly positive} entries for all $x \in \mathbb{R}$.
By Corollary~\ref{phiL1} and Lemma~\ref{Phiproperties}\,\ref{matrixelements}, the components of $\Phi_{\vec G}(x)$ are integrable, and so we can choose $\Psi_{\vec G} \in BC(\mathbb{R},\mathrm{Mat}(N,\mathbb{R}))$ such that  $\Psi_{\vec G}(x)$ has positive entries bounded away from zero and satisfies $\tfrac{d}{dx} \Psi_{\vec G}(x) = \Phi_{\vec G}(x)$. Comparing to Remark~\ref{remark-S-matrix-relation}, we see that with the above assumption on $\vec G$, it is always possible to find an
 $\vec S \in BC(\mathbb{R},\mathrm{Mat}(N,\mathbb{C}))$
such that \eqref{PhiG-logS} holds.
\end{remark}

\subsection{Convolution integrals involving $\phi_d(z)$}
\label{section:phiconvolution}

In this section we adopt again the convention \eqref{d_in_(-2,2)} that the parameter $d$ will always take values in the range
\be
	d\in(-2,2) \ .
\ee	 

Just as in the case of differential equations, the Green's function approach to difference equations will eventually express solutions in terms of convolution integrals involving the Green's function. For $g \in BC(\mathbb{R},\mathbb{C})$,
 the convolution with $\phi_d(z)$ is defined by 
 \be\label{Fd-convolution-def}
 F_d[g](z):= \int_{-\infty}^{\infty} \phi_d(z-t)g(t)\ dt \ .\ee 
 Due to Corollary~\ref{phiL1}, this function is well-defined on $\mathbb{S}_s$. As we will see in Section~\ref{GreenSoln}, it is important to understand the properties of such integrals as functions in $z$. That is the subject of this section.

The first question to ask is whether $F_d[g](z)$ is analytic. More generally: does the integration of a parameter-dependent analytic function preserve analyticity? The following lemma gives a criterion:

\begin{lemma}\label{analyticity}
Let $D\subseteq\mathbb{C}$ be a complex domain. Suppose $f:D\times \mathbb{R} \rightarrow \mathbb{C}$ is a function with the following properties:
\begin{enumerate}
\item \label{analyticity:1} for every $t_0\in\mathbb{R}$, the function $f(z,t_0)$ is analytic in $D$.
\item \label{analyticity:2} for every $z_0\in D$, the function $f(z_0,t)$ is continuous on $\mathbb{R}$.
\item \label{analyticity:3} for every $z_0\in D$ there exists a neighbourhood $U$ and an $L_1(\mathbb{R})$-integrable function $M(t)$, such that $\left|f(z,t)\right|\leq M(t)$ for all $z\in U$ and all $t\in \mathbb{R}$.
\end{enumerate}
Then the function 
\be
F(z) = \int_{-\infty}^\infty f(z,t)\; dt
\ee 
is analytic in $D$.
\end{lemma}

\begin{proof}
Let $z\in D$. Note that $F(z)$ is well-defined since the integrand is continuous (condition \ref{analyticity:2}) and dominated by an integrable function (condition \ref{analyticity:3}). Now take an arbitrary closed triangular contour $\Gamma$ inside $D$. 
Define the function 
\be
L(z):= \int_{-\infty}^\infty \left|f(z,t)\right| \;dt \ . \ee 
Since by condition \ref{analyticity:3}, $f(z,t)$ is locally dominated by an integrable function, $L(z)$ is continuous on $D$. Thus,  on the compact contour $\Gamma$ the function $L(z)$ is bounded and the integral \be\oint_\Gamma\left(\int_{-\infty}^\infty \left|f(z,t)\right|dt \right) dz\ee is finite. This warrants the application of Fubini's theorem, followed by analyticity (condition \ref{analyticity:1}): \be\oint_\Gamma F(z) dz = \oint_\Gamma\int_{-\infty}^\infty f(z,t)dt \; dz = \int_{-\infty}^\infty\oint_\Gamma f(z,t) dz\; dt = 0\ .\ee
By Morera's theorem, the claim follows.
\end{proof}

This lemma can be applied to the convolution integral $F_d[g](z)$. Set $D=\mathbb{S}_s$, and for any given $z_0\in\mathbb{S}_s$ set $U=B_\varepsilon(z_0)$
(the open ball with radius $\varepsilon$ and center $z_0$)
 with some sufficiently small $\varepsilon$. By Lemma \ref{phibounded}, a dominating integrable function $M(t)$ can be found by taking it to be a constant $B>0$ for $|t-z_0|<T$ and $B|t-z_0|^{-2}$ else. $B$ and $T$ are to be chosen sufficiently large. We have shown:

\begin{corollary}\label{PhiConvgAnalytic}
For every $g\in BC(\mathbb{R},\mathbb{C})$, the function $F_d[g](z)$ is analytic in $\mathbb{S}_s$.
\end{corollary}

Note that in the same fashion as for $\phi_d(z-t)g(t)$, one can also use Lemma \ref{phibounded} to construct integrable dominating functions for $\frac{d^n}{dz^n}\phi_d(z-t)g(t)$. Hence, we are allowed to differentiate inside the integral:
\be\label{Fswapdiffint}
\frac{d^n}{dz^n}F_d[g](z) = \int_{-\infty}^{\infty} \frac{d^n}{dz^n}\phi_d(z-t) \ g(t) \ dt \ .
\ee
More can be said about the nature of $F_d[g](z)$ and its derivatives. The following lemma provides a stepping stone.

\begin{lemma} \label{L1phiBounded}
Let $n\in\mathbb{Z}_{\geq 0}$. The function $y\mapsto\left\|\partial^n\phi_d^{[y]}\right\|_{L_1}$ is bounded in the compact interval $[-Y,Y]$ for every $0<Y<s$.
\end{lemma}

\begin{proof}
The function is well-defined due to Corollary \ref{phiL1}. Now fix a $\Theta\in[0,\frac{\pi}{2})$ and $0<Y<s$. Due to Lemma \ref{phibounded}, there exists a $C>0$ such that 
\be
\left|(x+iy)^2\partial^n\phi_d^{[y]}(x)\right|\leq C\ee for all
$x,y\in\mathbb{R}$
with $|\frac{y}{x}|<\tan\Theta$, and consequently \be\left|\partial^n\phi_d^{[y]}(x)\right|\leq \begin{cases} 
	\frac{C}{|x+iy|^2}\leq \frac{C}{x^2} & 
\mathrm{for}\;|x|>\frac{Y}{\tan\Theta}\\ \max_{z\in\mathbb{S}_Y}|\partial^n\phi_d(z)| & \mathrm{else}\end{cases} \ee for all $y\in[-Y,Y]$. The right hand side is in $L_1(\mathbb{R})$ and independent of $y$. Its integral over $\mathbb{R}$ provides a bound for $\left\|\partial^n\phi_d^{[y]}\right\|_{L_1}$ in the interval $[-Y,Y]$.
\end{proof}

\begin{lemma} \label{DerivConvBounded}
Let $n\in\mathbb{Z}_{\geq 0}$.
For every $g\in BC(\mathbb{R},\mathbb{C})$, the function $\frac{d^n}{dz^n}F_d[g](z)$ is bounded in $\mathbb{S}_Y$ for all $0<Y<s$.
\end{lemma}

\begin{proof}
With \ref{Fswapdiffint}, one has \be\left|\frac{d^n}{dz^n}F_d[g](z)\right|\leq \int_{-\infty}^{\infty} \left|\frac{d^n}{dz^n}\phi_d(z-t)\right|\left|g(t)\right| dt \leq \sup_{t\in\mathbb{R}}|g(t)|\left\|\partial^n\phi_d^{[\mathrm{Im}(z)]}\right\|_{L_1} \ .\ee According to Lemma \ref{L1phiBounded}, the right-hand side is bounded for $|\mathrm{Im}(z)|\leq Y$. Hence, $F_d[g](z)$ is bounded in $\mathbb{S}_Y$.
\end{proof}

Since $\phi_d(z)$ has poles in $z=\pm is$, there is no obvious way to extend the domain of $F_d[g](z)$ beyond $\mathbb{S}_s$. Lemma~\ref{analyticity} thus provides no information regarding the behaviour of this convolution integral as $z$ approaches the boundary $\partial\mathbb{S}_s=\mathbb{R}\pm is$. Moreover, Lemma \ref{DerivConvBounded} can only be used to prove boundedness of $F_d[g](z)$ in a strip which is strictly contained in $\mathbb{S}_s$. In the remainder of this section, we will show that for $g$ is H\"older continuous, $F_d[g](z)$ can be extended to $\overline{\mathbb{S}}_s$ as a bounded and continuous function. To this end, we need another result from complex analysis.

Let us relax the analyticity condition in Lemma \ref{analyticity}: suppose $f(z,t)$ is analytic everywhere except in $z=t$, where it shall have a pole of first order. Consider a contour $\gamma$ in $D$, and integrate over it: \be
F(z) = \int_\gamma f(z,t) \, dt \ .\ee 
The pole of $f(z,t)$ at $z=t$ causes $F(z)$ to have a branch cut along $\gamma$. Theorems describing this behaviour often go by the name of Sokhotski-Plemelji \cite{Gakhov}. 
The next proposition is an instance of this for $\gamma=\mathbb{R}$, and it follows from a more more general statement proven in Appendix~\ref{AppendixSokhotski}. 

\begin{proposition}\label{SokhotskiConvolutionMain}
Let $a>0$ and let $h:\mathbb{S}_a \rightarrow \mathbb{C}$ be an analytic function such that both $zh(z)$ and $\frac{d}{dz}h(z)$ are bounded in $\mathbb{S}_a$. 
Moreover, let $g:\mathbb{R}\rightarrow \mathbb{C}$ be a bounded 
H\"older continuous 
function. Then  
\be
F(z)=\int_{-\infty}^{\infty}\frac{h(z-t)}{z-t}g(t)\,dt\ee is analytic in $\mathbb{S}_a \setminus \mathbb{R}$. Moreover, the limits 
\be
F^\pm(x):=\lim_{y\searrow  s}F(x\pm iy) \hspace{1.5cm}(x\in\mathbb{R})\ee exist and are uniform in $x$.
 The functions $F^\pm(x)$ are bounded on $\mathbb{R}$ and provide  continuous extensions of $F(z)$ from the upper/lower half-plane to the real axis, 
related by 
\be
F^+(x)-F^-(x)~=~2i\pi \, h(0) g(x)\ .\ee
\end{proposition}

Now let us apply this result to convolution integrals involving $\phi_d(z)$: let $a<s$ and define the functions $h_\pm(z)=z\phi_d(z\pm is)$. By Lemma \ref{phianalyticstructure}, $h_\pm(z)$ are analytic in $\mathbb{S}_a$. Due to Lemma \ref{phibounded}, both $zh_\pm(z)$ and $\frac{d}{dz}h_\pm(z)$ are bounded in $\mathbb{S}_a$. Hence, $h_\pm(z)$ satisfy the conditions of Proposition \ref{SokhotskiConvolutionMain} for any $a<s$. Specifically, this clarifies the behaviour of our convolution integral as we approach the boundary of the strip:

\begin{corollary}\label{PhiConvgExtension}
Let $g:\mathbb{R}\rightarrow\mathbb{C}$ be bounded and H\"older continuous. Then the function $F_d[g](z)$ has a continuous extension to $\overline{\mathbb{S}}_s$, and this extension is bounded on $\partial\mathbb{S}_s=\mathbb{R}\pm is$.
\end{corollary}

Lastly, combining this result with the case $n=0$ of Lemma \ref{DerivConvBounded}, we obtain:

\begin{lemma} \label{PhiConvgBounded}
Let $g:\mathbb{R}\rightarrow\mathbb{C}$ be bounded and H\"older continuous. Then the function $F_d[g](z)$ is bounded in $\mathbb{S}_s$.
\end{lemma}

\begin{proof}
By Corollary \ref{PhiConvgExtension}, there is some constant $B>0$ such that $|F_d[g](z)|\leq B$ for all $z\in\partial\mathbb{S}_s$. According to Proposition \ref{SokhotskiConvolutionMain}, $F_d[g](x\pm iy)\rightarrow F_d[g](x\pm is)$ uniformly as $y\nearrow s$. Thus, for any given $\varepsilon >0$ there exists a $\delta>0$ such that $|F_d[g](z)|\leq B+\varepsilon$ for all $z\in\mathbb{S}_s\setminus\mathbb{S}_{s-\delta}$. In other words, $F_d[g](z)$ is bounded in $\mathbb{S}_s\setminus\mathbb{S}_{s-\delta}$. But on the other hand, by Lemma \ref{DerivConvBounded} $F_d[g](z)$ is also bounded in $\mathbb{S}_{s-\delta}$.
\end{proof}

The results of this section can now be summed up as follows: if $g$ is bounded and H\"older continuous, then $F_d[g]\in\mathcal{BA}(\mathbb{S}_s)$.

\subsection{Proof of Proposition \ref{FuncRelToNLIE}}
\label{GreenSoln}

$\ref{ffg-funrel} \Rightarrow \ref{ffg-inteq}$: For $y\in\mathbb{R}$, define the family of continuous functions $\vec{f}^{[y]}(x):=\vec{f}(x+iy)$. Continuity of $\vec{f}$ on the closure of the strip $\mathbb{S}_s$ guarantees pointwise convergence $\vec{f}^{[y]}\rightarrow \vec{f}^{[s]}$ as $y\nearrow s$. 
By boundedness of $\vec{f}$, the components $f^{[y]}_m$ of $\vec{f}^{[y]}$ are uniformly bounded by some constant $M$. It follows that, for any fixed value of $b\in\mathbb{R}$ and any	$d \in (-2,2)$, 
\be\left|\phi_d(b-x) 
f_m^{[y]}(x)\right|\leq M|\phi_d(b-x)| 
\qquad \text{for all}~~ x\in\mathbb{R}
\ .\ee The function on the right-hand-side is in $L_1(\mathbb{R})$ according to Corollary \ref{phiL1}. Thus, by Lebesgue's dominated convergence theorem we can write
\be\label{lebesguelimes}\int_{-\infty}^{\infty}\phi_d(x-t)f_m^{[s]}(t)dt = \lim_{y\nearrow  s}\int_{-\infty}^{\infty}\phi_d(x-t)f_{m}^{[y]}(t)dt \ .\ee 
By a simple change of variables followed by contour deformation (which is now allowed because for $y<s$ the contour lies inside the analytic domain) one can transfer the appearance of $y$ from $f_m$ to $\phi_d$:  \be\label{contdef}\int_{-\infty}^{\infty}\phi_d(x-t)f_m(t+iy)dt = \int_{-\infty}^{\infty}\phi_d(x+iy-t)f_m(t) dt \ee
Note that the integrals over the vertical parts of the contour vanish when pushed to infinity (see Corollary \ref{boundaryintegrals}). Plugging \eqref{contdef} into \eqref{lebesguelimes}, and making use of Lemma \ref{Phiproperties}\,\ref{matrixelements} to write $[\Phi_{\vec C}]_{nm}$
 in terms of the one-dimensional Green's functions $\phi_d$, gives rise to the identity \be\label{shiftidentity}\Phi_{\vec C}\star\vec{f}^{[\pm s]}(x) = \lim_{y\nearrow  s} \Phi_{\vec C}^{[\pm y]}\star\vec{f}(x)
\qquad \text{for all}~~
  x\in\mathbb{R}
   \ .\ee Taking into account $[\vec{C},\Phi_{\vec C}(x)]=0$ due to Lemma \ref{Phiproperties}\,\ref{simdiag} and distributivity of the convolution, \eqref{shiftidentity} directly implies \be\Phi_{\vec C}\star\left(\vec{f}^{[+s]}+\vec{f}^{[-s]}-(\vec{C}\cdot\vec{f})\right)(x) =\lim_{y\nearrow  s}\left(\Phi_{\vec C}^{[+y]}+\Phi_{\vec C}^{[+y]}-\vec{C}\cdot\Phi_{\vec C}\right)\star\vec{f}(x)\ee for all $x\in\mathbb{R}$. On the left-hand-side we can substitute the functional relation \eqref{ffuncrel}, and on the right-hand-side apply Lemma~\ref{Phiproperties}\,\ref{delta}. This results in \eqref{fconv}.

\medskip

\noindent
$\ref{ffg-inteq} \Rightarrow \ref{ffg-funrel}$: According to Lemma \ref{Phiproperties}\,\ref{matrixelements}, the components of $\vec{f}(x)=\left(\Phi_{\vec C}\star \vec{g}\right)(x)$ are given by real linear combinations of the form 
\be\label{phi*g-via-Fd}
\sum_{d,m} c_{d,m}F_d[g_m](x) \ ,
\ee 
where the $g_m(x)$ are bounded and H\"older continuous by assumption and
$F_d[g_m](x)$ are the convolution integrals discussed in Section~\ref{section:phiconvolution}. According to
Corollary~\ref{PhiConvgAnalytic},
 Corollary~\ref{PhiConvgExtension} and Lemma~\ref{PhiConvgBounded}, these integrals have
analytic continuations  $F_d[g_m]\in\mathcal{BA}(\mathbb{S}_s)$.
This shows that $\vec{f}\in \mathcal{BA}(\mathbb{S}_s)^N$
and
\be\label{fy-as-convolution-with-g}
	\vec{f}^{[y]}(x)=\left(\Phi_{\vec C}^{[y]}\star \vec{g}\right)(x) \ .
\ee
To obtain the functional equation \eqref{ffuncrel} we basically reverse the above reasoning,
\begin{align}
\left(\vec{f}^{[+s]}+\vec{f}^{[-s]}-(\vec{C}\cdot\vec{f})\right)(x)
&\overset{(a)}=
\lim_{y\nearrow  s}
\left(\vec{f}^{[+y]}+\vec{f}^{[-y]}-(\vec{C}\cdot\vec{f})\right)(x)
\nonumber\\
&\overset{\eqref{fy-as-convolution-with-g}}=
\lim_{y\nearrow  s}
\left(\left(\Phi_{\vec C}^{[+y]}+\Phi_{\vec C}^{[-y]}-(\vec{C}\cdot\Phi_{\vec C})\right)\star \vec{g}\right)(x)
\nonumber\\
&\overset{(b)}=\vec{g}(x) \ .
\end{align}
Here (a) follows from pointwise convergence $\vec{f}^{[\pm y]}(x) \to \vec{f}^{[\pm s]}(x)$ due to continuity of $\vec{f}$ on $\overline{\mathbb{S}}_s$, and step (b) is the $\delta$-function property from Lemma~\ref{Phiproperties}\,\ref{delta}.

This completes the proof of Proposition~\ref{FuncRelToNLIE}.

\section{Unique solution to a family of integral equations}
\label{section-unique-integral-soln}

In this section we give a criterion for functional equations of the form \eqref{TBAintro} in the introduction to have a unique solution. 
Specifically, we will prove a special case of Theorem~\ref{main-theorem-TBA} where we choose $\vec{C} = \tfrac12 \vec{G}$.

\begin{proposition}\label{TBAuniquenessDynkin} Let $\vec G \in \mathrm{Mat}_{<2}(N)$ be non-negative and irreducible,
and let $\vec{a} \in BC_-(\mathbb{R},\mathbb{R})^N$.
Then the system of nonlinear integral equations
\be\label{Csystem}
\vec{f}(x)=\int_{-\infty}^{\infty} \Phi_{\frac{1}{2}\vec G}(x-y)\cdot\vec{G}\cdot\left(\log\left(e^{-\vec{a}(y)}+e^{\vec{f}(y)}\right)-
	\tfrac{1}{2}\vec{f}(y)
\right) dy
\ee
has exactly one bounded continuous solution, 
$\vec{f}_\star\in BC(\mathbb{R},\mathbb{R})^N$.
\end{proposition}

The proof will consist of verifying that the Banach Fixed Point Theorem can be applied and will be given in Section~\ref{uniqunessproofTBA}. In Sections~\ref{hammerstein-section} and~\ref{section:TBAuniqueness} we lay the groundwork by discussing a type of integral equations called Hammerstein integral equations, and by applying the general results there to TBA-type equations. 

After the proof, in Section~\ref{G-is-graph-section} we comment on the special case that $\vec{G}$ is the adjacency matrix of a graph -- the case most commonly considered in applications -- and in Section~\ref{N=1-example-section} we look at the case $N=1$ in more detail.

The only previous proof of Theorem~\ref{main-theorem-TBA} we know of 
concerns the case 
	$N=1$, $\vec{G}=\vec{C}=1$ and $\vec{a} \sim \cosh(x)$, and
can be found in \cite{FringKorffSchulz}.
Their argument 
	also uses the Banach Fixed Point Theorem but is different from ours (we rely on being able to choose $\vec{C}$ different from $\vec{G}$) and we review it in Section~\ref{N=1-example-section}.

\subsection{Hammerstein integral equations as contractions}
\label{hammerstein-section}

In this section we take $\mathbb{K}$ to stand for $\mathbb{R}$ or $\mathbb{C}$. We use the abbreviation
$BC(\mathbb{R}) := BC(\mathbb{R},\mathbb{K})$. Similarly, we write $BC(\mathbb{R}^m)^N$ for $BC(\mathbb{R}^m,\mathbb{K}^N)$, which we think of either as $\mathbb{K}^N$-valued functions, or as $N$-tuples of $\mathbb{K}$-valued functions.

\medskip

Consider the nonlinear integral equation 
\be
f(x) = \int_{-\infty}^{\infty} K(x,y)L(y,f(y)) dy \ ,\ee where 
	$K:\mathbb{R}\times\mathbb{R}\rightarrow\mathbb{K}$ and
	$L:\mathbb{R}\times\mathbb{K}\rightarrow\mathbb{K}$ 
are some functions continuous in both arguments, and where it is understood that the integral is well-defined for  $f(x)$ in some suitable class of functions on $\mathbb{R}$. Integral equations of this form are commonly referred to as \textsl{Hammerstein equations}, see e.g.\ \cite[I.3]{Krasnoselskii}
	and \cite[Ch.\,16]{Polyanin}.
A function $f(x)$ solves this equation if and only if it is a fixpoint of the corresponding integral operator 
\be
A[f](x) := \int_{-\infty}^{\infty} K(x,y)L(y,f(y)) dy \ .\ee When does such a map have a unique fixpoint? We will try to bring Banach's Fixed Point Theorem to bear on this question, which we now briefly recall.

\begin{definition} 
Let $X$ be a metric space. A map $A:X\rightarrow X$ is called a {\sl contraction} if there exists a positive real constant $\kappa<1$ such that \be
d_X(A(x),A(y)) \leq \kappa \, d_X(x,y)\ee
for all $x,y\in X$. If the condition is satisfied for $\kappa=1$, then $A$ is called {\sl non-expansive}.
\end{definition}

\begin{theorem}[Banach]\label{Banach}
Let $X$ be a complete metric space and $A:X\rightarrow X$ a contraction. Then $A$ has a unique fixpoint $x_\star\in X$. Furthermore, for every $x_0\in X$, the recursive sequence $x_n:=A(x_{n-1})$ converges to $x_\star$.
\end{theorem}

We now describe a general principle
which facilitates the application of 
Banach's Theorem 
to integral operators of Hammerstein type 
	(see e.g.\ \cite[16.6-1,\,Thm.\,3]{Polyanin}).
Suppose there exists a constant $\rho>0$ such that 
\be\int_{-\infty}^{\infty} \left|K(x,y)\right| dy \leq \rho
	\qquad \text{for all} \quad
x\in\mathbb{R} \ .
\ee 
Suppose also that $L$ is Lipschitz continuous in the second variable, 
i.e.\ there exists a constant $\sigma>0$ such that
\be\left|L(x,t)-L(x,s)\right| \leq \sigma \left|t-s\right| 
	\qquad \text{for all} \quad
x\in\mathbb{R}, \ s,t\in\mathbb{K}\ .\ee
Provided $A[f]$ defines a map $BC(\mathbb{R})\rightarrow BC(\mathbb{R})$, 
we can use this to compute for any $f,g\in BC(\mathbb{R})$:
\begin{align}
\left\| A[f]-A[g] \right\|_\infty &= \sup_{x\in\mathbb{R}} \left| \int_{-\infty}^{\infty} K(x,t) \left( L(t,f(t))-L(t,g(t)) \right) dt \right| \nonumber\\
&\leq \sup_{x\in\mathbb{R}}  \int_{-\infty}^{\infty}\left| K(x,t)\right| \left| L(t,f(t))-L(t,g(t))  \right| dt  \nonumber\\
&\leq \sigma \, \sup_{x\in\mathbb{R}}  \int_{-\infty}^{\infty}\left| K(x,t)\right| \left| f(t)-g(t)  \right| dt 
~\leq~ \sigma\rho \,\left\| f-g \right\|_\infty
\end{align}
If $\kappa:=\sigma\rho<1$, then $A[f]$ is a contraction with respect to the metric induced by the supremum norm $\left\|\cdot\right\|_\infty$. Recall that $BC(\mathbb{R})$ together with the norm $\left\|\cdot\right\|_\infty$ is a Banach space, and so the Banach Theorem~\ref{Banach} applies.

Consider now $N$ coupled nonlinear integral equations of Hammerstein type:
\be\vec{f}(x) = \int_{-\infty}^{\infty} \vec{K}(x,y)\cdot\vec{L}(y,\vec{f}(y)) \, dy \ ,\ee 
where
	$\vec{K}:\mathbb{R}\times\mathbb{R}\rightarrow \mathrm{Mat}(N,\mathbb{K})$ and
	$\vec{L}:\mathbb{R}\times\mathbb{K}^N\rightarrow\mathbb{K}^N$ 
are continuous in both arguments.
Our arguments depend crucially on the right choice of norm for the functions $\vec{f} :\mathbb{R} \to \mathbb{K}^N$.

For $1\leq p \leq\infty$, we equip the space $BC(\mathbb{R})^N$ with the norm $\left\|\cdot\right\|_{\infty_p}$ given by \be\left\|\vec{f}\right\|_{\infty_p} := \sup_{x\in\mathbb{R}}\left\|\vec{f}(x)\right\|_p = \sup_{x\in\mathbb{R}}\left(\sum_{i=1}^N |f_i(x)|^p\right)^\frac{1}{p} \ .\ee
The normed space $\left(BC(\mathbb{R})^N, \left\|\cdot\right\|_{\infty_p}\right)$  is a Banach space by the following standard lemma, which we state without proof.

\begin{lemma} \label{B(X,Y)Banach} Let $X$ be a metric space and $(Y,\left\|\cdot\right\|_Y)$ a Banach space. 
Then also $\left(BC(X,Y),\left\|\cdot\right\|_\infty\right)$ is a Banach space.
\end{lemma}

We now give the $N$-component version of the general principle outlined above. The proof is the same, just with heavier notation.

\begin{lemma} \label{HammersteinContractionNdim} Let $\vec{K}\in\mathrm{Mat}(N,BC(\mathbb{R}^2))$ and $\vec{L}:\mathbb{R}\times\mathbb{K}^N \rightarrow\mathbb{K}^N$ a function. Suppose that 
\begin{itemize}
\item 
$\vec{K}$ has bounded integrals in the second variable, in the sense that
\be
\rho_{ij}:= \sup_{x\in\mathbb{R}}\int_{-\infty}^{\infty} \left|K_{ij}(x,y)\right| dy  \,<\, \infty  \qquad , \quad i,j=1,...,N \ .
\ee
\item
all the components of $\vec{L}$ are Lipschitz continuous in the second variable in the sense that there are constants $\sigma_j\geq 0$, 
$j=1,\dots,N$, such that
\be
\left|L_j(y,\vec{v})-L_j(y,\vec{w})\right| \leq \sigma_j \left\|\vec{v}-\vec{w}\right\|_p \quad \text{for all} ~~
\vec{v},\vec{w}\in
	\mathbb{K}^N, 
\; y\in\mathbb{R}  \ .\ee 
\item
the matrix $\boldsymbol\rho=(\rho_{ij})_{i,j=1,...,N}$ and the vector $\boldsymbol\sigma=(\sigma_i)_{i=1,...,N}$ are
such that $\kappa := \left\|\boldsymbol\rho\cdot\boldsymbol\sigma\right\|_p < 1$. 
\end{itemize}
Suppose that the following
integral operator defines a map 
$\vec{A}:BC(\mathbb{R})^N \rightarrow BC(\mathbb{R})^N$,
\be\vec{A}[\vec{f}](x) := \int_{-\infty}^{\infty} \vec{K}(x,y)\cdot\vec{L}(y,\vec{f}(y)) \, dy \ .
\ee
Then $\vec{A}$ is a contraction
on $\left(BC(\mathbb{R})^N, \left\|\cdot\right\|_{\infty_p}\right)$.
\end{lemma}

\begin{proof} Let $\vec{f},\vec{g}\in BC(\mathbb{R})^N$. Then
\allowdisplaybreaks
\begin{align}
\left(\left\| \vec{A}[\vec{f}]-\vec{A}[\vec{g}] \right\|_{\infty_p}\right)^p &= \sup_{x\in\mathbb{R}} \sum_{i=1}^N \left|A_i[\vec{f}]-A_i[\vec{g}]\right|^p \nonumber\\
&= \sup_{x\in\mathbb{R}} \sum_{i=1}^N \left|\sum_{j=1}^N\int_{-\infty}^{\infty} K_{ij}(x,y)\left(L_j(y,\vec{f}(y))-L_j(y,\vec{g}(y))\right) dy\right|^p \nonumber\\
&\leq \sup_{x\in\mathbb{R}} \sum_{i=1}^N \left|\sum_{j=1}^N\int_{-\infty}^{\infty} \left|K_{ij}(x,y)\right|\big|L_j(y,\vec{f}(y))-L_j(y,\vec{g}(y))\big| dy\right|^p \nonumber\\
&\leq \left(\left\|\vec{f}-\vec{g}\right\|_{\infty_p}\right)^p \,\sum_{i=1}^N \left|\sum_{j=1}^N \rho_{ij}\sigma_j\right|^p
 \nonumber \\
&= \left(\left\|\vec{f}-\vec{g}\right\|_{\infty_p}\right)^p \left(\left\|\boldsymbol\rho\cdot\boldsymbol\sigma\right\|_p\right)^p 
=\left(\kappa \left\|\vec{f}-\vec{g}\right\|_{\infty_p}\right)^p  \ .
\end{align}
Since $\kappa<1$, $\vec{A}$ is a contraction.
\end{proof}

\subsection{Unique solution to TBA-type equations}
\label{section:TBAuniqueness}

In this section we specialise the results of the previous section to integral equations of the form \eqref{TBAintro}. We will restrict ourselves to the case $\mathbb{K} = \mathbb{R}$.

Let us call a function $f:\mathbb{R}^N\rightarrow\mathbb{R}$ \textsl{$p$-Lipschitz-continuous}, if it satisfies the Lipschitz condition with respect to $\left\|\cdot\right\|_p$. In fact, this is a slightly redundant denomination: equivalence of all $p$-norms ensures that $f$ is Lipschitz continuous either with respect to all or none of the $p$-norms. However, the optimal Lipschitz constants differ, which is important in view of the third condition in 
	Lemma~\ref{HammersteinContractionNdim}.
 For differentiable functions that satisfy a Lipschitz condition, we can characterise the $p$-Lipschitz constant in terms of the gradient as follows:

\begin{lemma}\label{pLipschitz}
Let $f:\mathbb{R}^N\rightarrow\mathbb{R}$ be a continuous function whose gradient $\nabla f:\mathbb{R}^N\rightarrow\mathbb{R}^N$ is also a continuous function, and $1\leq p \leq\infty$. Then the following inequality holds:
\be\left|f(\vec{v})-f(\vec{w})\right| \leq \left(\sup_{\vec{u}\in\mathbb{R}^N}\left\|\nabla f(\vec{u})\right\|_q\right) \left\|\vec{v}-\vec{w}\right\|_p \hspace{1cm} \ \text{for all}~~ \vec{v},\vec{w}\in\mathbb{R}^N\ ,\ee where $q$ is defined via the relation $\frac{1}{p}+\frac{1}{q}=1$. In particular, $f$ is Lipschitz continuous if the gradient is bounded.
\end{lemma}

\begin{proof}
Let $\vec{v},\vec{w}\in\mathbb{R}^N$. By the mean value theorem there exists a $t\in[0,1]$ such that \be\left|f(\vec{v})-f(\vec{w})\right| = \left|\nabla f(\vec{x})\cdot(\vec{v}-\vec{w})\right|\ee for $\vec{x}=\vec{v}+t(\vec{w}-\vec{v})$. Applying first the triangle and then the H\"older inequality to the right hand side, one obtains
\begin{align}
\left|f(\vec{v})-f(\vec{w})\right| &\leq  \sum_{i=1}^N\left|\nabla_i f(\vec{x})(v_i-w_i)\right| \nonumber\\
&\leq \left\|\nabla f(\vec{x})\right\|_q \left\|\vec{v}-\vec{w}\right\|_p \nonumber\\
&\leq \left(\sup_{\vec{u}\in\mathbb{R}^N}\left\|\nabla f(\vec{u})\right\|_q\right) \left\|\vec{v}-\vec{w}\right\|_p \ ,
\end{align}
which means that $\sup_{\vec{u}\in\mathbb{R}^N}\left\|\nabla f(\vec{u})\right\|_q$ , if bounded, is a $p$-Lipschitz constant for $f$. It is easy to see that it is the optimal $p$-Lipschitz constant.
\end{proof}

\begin{lemma} \label{StandardLfunctionIsLipschitz}
Let $A\ge 0$ and $a,b,c\in\mathbb{R}$. The function 
$L:\mathbb{R}\rightarrow\mathbb{R}$ defined by 
\be
L(x)=a\cdot\log\left(A+e^{bx}\right)-cx
\ee is Lipschitz continuous with Lipschitz constant $\sigma_L=\max(|c|,|ab-c|)$.
\end{lemma}

\begin{proof}
If any one of $a,b,A$ is zero, the statement is clear. Suppose $a,b,A\neq0$.
The derivative \be\frac{d}{dx}L(x)=\frac{ab}{Ae^{-bx}+1}-c\ee interpolates between $-c$ (for $bx\rightarrow -\infty$) and $ab-c$ (for $bx\rightarrow \infty$). It is also a monotonous function, since 
\be\frac{d^2}{dx^2}L(x)= a \cdot \frac{Ab^2 e^{bx}}{(A+e^{bx})^2}\ ,
\ee
which is $\ge 0$ 
for $a \ge 0$ and $\le 0$ for $a \le 0$.
It thus follows that $\left|\frac{d}{dx}L(x)\right|\leq\max(|c|,|ab-c|)$.
Lemma~\ref{pLipschitz} completes the proof.
\end{proof}

\begin{proposition} \label{UniqueGroundstateNdim} Let $1\leq p,q \leq\infty$ such that $\frac{1}{p}+\frac{1}{q}=1$ holds. For $i,j=1,...,N$, let $\phi_{ij}\in 
	BC(\mathbb{R},\mathbb{R})
\cap L_1(\mathbb{R})$, $a_j\in BC_-(\mathbb{R},\mathbb{R})$ 
and $G_{ij},C_{ij}\in\mathbb{R}$, $w_i\in\mathbb{R}_{>0}$. Furthermore, set $M_{ij}:=\max\left(|C_{ij}|,|G_{ij}-C_{ij}|\right)$ and define \be \sigma_i:=\left( \sum_{j=1}^N (M_{ij}w_j)^q \right)^\frac{1}{q} 
\quad , \qquad
\rho_{ij}:=\frac{1}{w_i}\int_{-\infty}^\infty\left|\phi_{ij}(y)\right|dy \ .\ee If 
$\kappa:=\left\|\boldsymbol\rho\cdot\boldsymbol\sigma\right\|_p<1$, then the system of nonlinear integral equations given by 
\be\label{TBA-type-contracting-integral}
f_i(x)=\sum_{j=1}^N\int_{-\infty}^{\infty} \phi_{ij}(x-y)\sum_{k=1}^N\left[G_{jk}\log\left(e^{-a_k(y)}+e^{f_k(y)}\right)-C_{jk}f_k(y)\right] dy
\ee 
has exactly one bounded continuous solution, 
$(f_{\star,1},\dots,f_{\star,N})\in BC(\mathbb{R},\mathbb{R})^N$.
\end{proposition}

\begin{proof}
We start by rewriting \eqref{TBA-type-contracting-integral} in terms of the rescaled functions $g_i(x)= f_i(x)/w_i$: 
\be\label{contractingintegralforg}
g_i(x)=\sum_{j=1}^N\int_{-\infty}^{\infty} \frac{1}{w_i}\phi_{ij}(x-y)\sum_{k=1}^N\left[G_{jk}\log\left(e^{-a_k(y)}+e^{w_k g_k(y)}\right)-C_{jk}w_k g_k(y)\right] dy
\ee 
The lower bound on the $a_k(y)$ ensures that the corresponding integral operator is a well-defined map $BC(\mathbb{R},\mathbb{R})^N\rightarrow BC(\mathbb{R},\mathbb{R})^N$. Consider the functions 
\be
L_i(y,\vec{v})=\sum_{k=1}^N\left[G_{ik}\log\left(e^{-a_k(y)}+e^{w_k v_k}\right)-C_{ik}w_k v_k\right] \ .\ee The individual summands, as functions in $v_k$ for fixed $y$,  are of the form considered in Lemma \ref{StandardLfunctionIsLipschitz}, from which one concludes \be\left|\frac{\partial}{\partial v_j} L_i(y,\vec{v})\right|\leq w_j\, \max\left(|C_{ij}|,|G_{ij}-C_{ij}|\right) \ .\ee According to Lemma \ref{pLipschitz}, the functions $L_i(y,\vec{v})$ are thus $p$-Lipschitz continuous in the second variable, with Lipschitz constants 
\begin{align}
\sigma_i &=  \sup_{\substack{\vec{v}\in\mathbb{R}^N \\ y\in\mathbb{R}}}\left\|\nabla_\vec{v} L_i(y,\vec{v})\right\|_q = \sup_{\substack{\vec{v}\in\mathbb{R}^N \\ y\in\mathbb{R}}}\left(\sum_{j=1}^N\left|\frac{\partial}{\partial v_j} L_i(y,\vec{v})\right|^q\right)^\frac{1}{q} \nonumber \\ &=\left(\sum_{j=1}^N\sup_{\substack{\vec{v}\in\mathbb{R}^N \\ y\in\mathbb{R}}}\left|\frac{\partial}{\partial v_j} L_i(y,\vec{v})\right|^q\right)^\frac{1}{q} =\left(\sum_{j=1}^N \left(w_j\, \max\left(|C_{ij}|,|G_{ij}-C_{ij}|\right)\right)^q\right)^\frac{1}{q} \ .
\end{align}
An application of 
	Lemma~\ref{HammersteinContractionNdim} 
completes the proof.
\end{proof}

We note that the bound on $\kappa$ used in Proposition~\ref{UniqueGroundstateNdim} is actually independent of the functions $a_j$.

\subsection{Proof of Proposition \ref{TBAuniquenessDynkin}}\label{uniqunessproofTBA}

The proof makes use of the Perron-Frobenius Theorem in the following form 
(see e.g.~\cite[Theorem 2.2.1]{BrouwerHaemers}):

\begin{theorem}[Perron-Frobenius]\label{PF}
Let $\vec{A}$ be a non-negative real-valued irreducible $N\times N$ matrix. Then the largest eigenvalue $\lambda_{\mathrm{PF}}$ of $\vec{A}$ is real and has geometric and algebraic multiplicity 1. Its associated eigenvector can be chosen to have strictly positive components
and is the only eigenvector with that property.
\end{theorem}

\medskip

We now turn to the proof of Proposition~\ref{TBAuniquenessDynkin}.

\medskip

Set $p=\infty$ and $q=1$. In terms of Proposition \ref{UniqueGroundstateNdim} we have $\vec{C}=\frac{1}{2}\vec{G}$, so that
$\vec{M}=\frac{1}{2}\vec{G}$. 
By the Perron-Frobenius theorem, $\vec{G}$ has an eigenvector $\vec{w}$ with strictly positive components $w_i>0$ associated to its largest eigenvalue $\lambda_{\mathrm{PF}}$. With this choice of $w_i$ the constant 
vector $\boldsymbol{\sigma}$ in 
Proposition~\ref{UniqueGroundstateNdim}
is given by
\be\label{sig}
\boldsymbol\sigma 
= \tfrac12 \vec{G}\vec{w}
= \tfrac12 \lambda_{\mathrm{PF}} \vec{w} \ .
\ee

Let us abbreviate $\Phi(x) := \Phi_{\frac{1}{2}\vec G}(x)$ with components $\phi_{ij}(x)$. 
Due to Lemma \ref{Phinonneg} we know that
$|\phi_{ij}(x)|=\phi_{ij}(x)$ 
for all $x\in\mathbb{R}$ and $i,j=1,\dots,N$. Combining this with \eqref{phiC-infinf-integral} one computes the matrix $\boldsymbol{\rho}$ in 
Proposition~\ref{UniqueGroundstateNdim} to be
\be
\rho_{ij} = \frac{1}{w_i} \int_{-\infty}^\infty\phi_{ij}(y)dy 
= 
\frac{1}{w_i}\big[ (2\,\mathbf{1} - \tfrac{1}{2}\vec{G})^{-1}\big]_{ij} \ .
\ee
Since $\vec{w}$ is an eigenvector of $(2\,\mathbf{1} - \tfrac{1}{2}\vec{G})^{-1}$ we find
\be\label{rhow}
\boldsymbol\rho \cdot \vec{w} = \frac{1}{2-\frac{\lambda_{\mathrm{PF}}}{2}} \cdot (1,1,\dots,1) \ .
\ee
Hence the contraction constant $\kappa$ in 
Proposition~\ref{UniqueGroundstateNdim} is given by
\be\kappa = \left\|\boldsymbol\rho\cdot\boldsymbol\sigma\right\|_\infty  = \frac{\lambda_{\mathrm{PF}}}{2}\max_{i=1,...,N}\left|\sum_{j=1}^N \rho_{ij}w_j\right| = \frac{\lambda_{\mathrm{PF}}}{2}\left|\frac{1}{2-\frac{\lambda_{\mathrm{PF}}}{2}}\right|=\left|\frac{\lambda_{\mathrm{PF}}}{4-\lambda_{\mathrm{PF}}}\right| \ .
\ee
It follows that $\kappa<1$ if and only if $\lambda_{\mathrm{PF}}<2$.

By Proposition~\ref{UniqueGroundstateNdim} there is a unique solution to \eqref{Csystem}, completing the proof of Proposition~\ref{TBAuniquenessDynkin}.

\subsection{Adjacency matrices of graphs}\label{G-is-graph-section}

In Proposition~\ref{TBAuniquenessDynkin}, the matrix $\vec{G}$ may have non-negative real entries. This is in itself interesting because it makes the solution
	$\vec{f}_\star$
depend on an additional set of continuous parameters.
However, in the application to integrable quantum field theory, $\vec{G}$ is usually the \textsl{adjacency matrix} of some (suitably generalised) graph.
Irreducibility is then equivalent to the corresponding generalised graph being strongly connected.

If $\vec{G}$ is symmetric and has entries $\lbrace 0,1\rbrace$ with zero on the diagonal, then by definition it is the adjacency matrix of a simple (undirected, unweighted) graph whose nodes $i$ and $j$ are connected if and only if $G_{ij}=1$. 
In this case, strongly connected and connected are equivalent. 
The only connected simple graphs with $\lambda_{\mathrm{PF}}<2$ are the graphs associated to the $ADE$ Dynkin diagrams
(while their affine versions are the sole examples satisfying $\lambda_{\mathrm{PF}}=2$), and their adjacency matrix is diagonalisable over $\mathbb{R}$ \cite[Thm.\,3.1.3]{BrouwerHaemers}.

Consider now generalised graphs. 
If we allow for loops ($G_{ii}\neq 0$, edges connecting a node to itself) and 
multiple edges between the same nodes (where the entry $G_{ij}$ is the number of edges connecting nodes $i$ and $j$)
we additionally get the tadpole
$T_N=A_{2N}/\mathbb{Z}_2$ (defined as the adjacency matrix of $A_N$ with additional entry $G_{11}=1$).
If, moreover, the symmetry requirement is dropped ($G_{ij}\neq G_{ji}$), we may still associate to $\vec{G}$ a mixed multigraph (with some edges now being replaced by arrows). The $BCFG$ Dynkin diagrams provide examples of this type for which $\lambda_{\mathrm{PF}}<2$ hold. However, in contrast with the undirected (symmetric) case, they are not exhaustive at all. For instance, a directed graph with $\lambda_{\mathrm{PF}}>2$ can be turned into a new directed graph with $\lambda_{\mathrm{PF}}<2$ by
subdividing all its edges often enough.\footnote{We are grateful to Nathan Bowler 
for pointing this out to us.} If $\vec{G}$ is not even assumed to be integer-valued, then every non-example becomes an example after appropriate rescaling.

To summarise, we have the following special cases in which Proposition \ref{TBAuniquenessDynkin} applies:

\begin{corollary} 
If $\vec{G}$ is the adjacency matrix of a finite
Dynkin diagram
($A_N$, $B_N$, $C_N$, $D_N$, $E_6$, $E_7$, $E_8$, $F_4$ or $G_2$)
or of the tadpole $T_N$, 
then the system \eqref{Csystem} has exactly one bounded continuous solution.
\end{corollary}

\subsection{The case of a  single TBA equation}\label{N=1-example-section}

Consider a single integral equation of TBA-type ($N=1$ with 
	$G_{11} = g\in(0,2)$, 
$C_{11} = c\in(-2,2)$): 
\be\label{TBAforN1}
f(x)=\int_{-\infty}^{\infty} \phi_c(x-y)\left[g\log\left(e^{-a(y)}+e^{f(y)}\right)-c f(y)\right] dy
\ee
This case is instructive: we do not have to worry about the choice of $p$ and $q$, nor does a rescaling $f(x)\rightarrow f(x)/w$ influence the contraction constant.
We compute the quantities in Proposition~\ref{UniqueGroundstateNdim} to be
\begin{align}
&\rho = \int_{-\infty}^\infty \phi_c(y)dy = \frac{1}{2-c}
\quad , \qquad
\sigma = M = \max(|c|, |g-c|) \ ,\nonumber\\
&\kappa(c) = \rho\sigma =\frac{\max(|c|, |g-c|)}{2-c} \ .
\end{align} 
The most important case is $g=1$ (this case arises for example in the Yang-Lee model). The contraction constant for this case is shown in Figure~\ref{kappaofc}.

\begin{figure}
\begin{center}
\includegraphics[width=10cm]{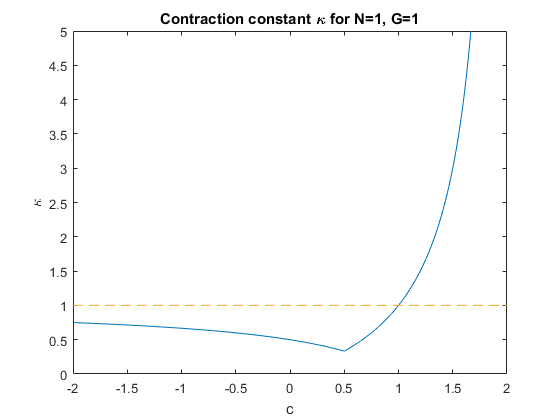}
\caption{Contraction constant for a single TBA equation at different values of $c$.}
\label{kappaofc}
\end{center}
\end{figure}

For $c<1$, our estimate guarantees a contraction. 
The canonical TBA equation with $c=g=1$ corresponds to the marginal case $\kappa=1$, whereas the universal TBA with $c=0$ yields $\kappa=\tfrac{1}{2}$. The ``sweet spot'' is $c=\tfrac{1}{2}=\tfrac{1}{2}g$ where our estimate for the contraction constant
attains its minimum $\kappa=\tfrac{1}{3}$.

If we chose a different value $g\in(0,2)$, the sweet spot shifts, but the overall picture remains the same (the region of assured contraction is $c<1$). However, for $g\geq 2$, $\kappa$ is larger or equal to one everywhere, and no region of assured contraction exists.

\begin{remark}\label{kappa-is-not-optimal}
It is noteworthy that $\kappa$ has virtually no practical bearing on the speed of convergence of the iterative numerical solution of equation \eqref{TBAforN1}. We solved it numerically for $g=1$, $a(x)=r\cosh(x)$ (the massive Yang-Lee model in volume $r$) for different values of $c\in[0,1]$ and $r\in(0,1]$, and the speed of convergence
(as measured by the number of iterations required to obtain a certain accuracy) increases almost linearly in $c$, instead of being governed by $\kappa$.
If $\kappa$ were the optimal contraction constant, we would instead expect to see the fastest convergence for $c=\tfrac12$.
\end{remark}

We now describe in more detail the proof given in 
\cite[Sec.\,5]{FringKorffSchulz}, which concerns the case $g=c=1$ and $a(x)=r\cosh(x)$ with $r>0$,
in the above example.
As described above, our bound on $\kappa$ in this case is $1$, so that Proposition~\ref{UniqueGroundstateNdim} does not apply. We circumvent this by using $c=\tfrac12$ and proving (in Section~\ref{section:uniquenessY}, see Theorem~\ref{main-theorem-TBA} from the introduction) that the fixed point is independent of $c$ and unique in the range
$c \in (-2,2)$.
The argument of \cite{FringKorffSchulz} also uses the Banach Theorem but proceeds differently.

Namely, they exhaust the space of bounded continuous functions by subspaces with specific bounds, 
	$BC(\mathbb{R},\mathbb{R})=\bigcup_{q \in [e^{-r},1)} D_{q,r}$,
where 
\be
D_{q,r}:=\Big\lbrace f\in BC(\mathbb{R},\mathbb{R})\,\Big| \left\|f\right\|_\infty \leq \log(\tfrac{q}{1-q})+r\Big\rbrace\ .
\ee 
Using convexity of $D_{q,r}$ it is shown that 
	for $q \ge e^{-r}$
the corresponding integral operator maps $D_{q,r}$ to itself, and that the associated contraction constant is $\kappa=q$.
Hence, as expected one obtains $\kappa\rightarrow 1$ when $q\rightarrow 1$ (so that $\log(\tfrac{q}{1-q}) \to \infty$), but on each $D_{q,r}$ we have $\kappa<1$ and thus a proper contraction. This implies a unique solution on the whole space $BC(\mathbb{R},\mathbb{R})$.

It is also stated in \cite{FringKorffSchulz} that the generalisation to higher $N$ should be straightforward. This is less clear to us. In the examples from finite Dynkin diagrams we need all the extra freedom we introduce in Proposition~\ref{UniqueGroundstateNdim} in an optimal way in order to press our bound $\kappa$ under 1. If we take, for example, $\vec G = \vec C$ to be the adjacency matrix of the $A_N$ Dynkin diagram, 
the quantities in Proposition~\ref{UniqueGroundstateNdim} take the values
\be
	M_{ij} = G_{ij} 
	~~,\quad
	\rho_{ij} = \left[(A_N)^{-1}\right]_{ij}
	~~,\quad
	\sigma_i = 
	\begin{cases} 1 &; i = 1,N \\ 2^{1/q} &; i = 2, \dots, N-1 \end{cases}
	\quad ,
\ee
where $A_N$  denotes the Cartan matrix $2 \mathbf{1} - \vec G$. With some more work, from this one can estimate that for $N$ large enough one has
$\left\|\boldsymbol\rho\cdot\boldsymbol\sigma\right\|_p > \tfrac18 N^2$ (independent of $p$ and $q$).
Hence, as opposed to what happened at $N=1$, for larger $N$ the bound $\kappa$ is not 1 but grows 
	at least
quadratically with $N$.
It is not obvious to us how to obtain drastically better estimates by choosing subsets analogous to $D_{q,r}$ of $BC(\mathbb{R},\mathbb{R}^N)$.

However, it might be possible -- at least in the massive case, cf.\ Remark~\ref{remark-constantY}\,\ref{spectralradius2isbad} -- that one can combine the freedom to choose $\vec C$ and $w_i$ that we introduce with the method of \cite{FringKorffSchulz} to extend our results to the case of 
spectral radius 2 or to 
infinitely many coupled TBA equations (such as the $N\to\infty$ limit of classical finite Dynkin diagrams, where the spectral radius approaches 2). The idea of choosing subsets $D_{q,r}$ as above might push $\kappa$ strictly below 1.
We hope to return to these points in the future.

\section{Uniqueness of solution to the Y-system}
\label{section:uniquenessY}

For $\vec{G}\in\mathrm{Mat}(N,\mathbb{R})$, $\vec{C}\in\mathrm{Mat}_{<2}(N)$
and $\vec{a}\in BC_-(\mathbb{R},\mathbb{R})^N$ we define the map $\vec{L}_{\vec{C}} : BC(\mathbb{R},\mathbb{R})^N \to BC(\mathbb{R},\mathbb{R})^N$ as (recall the convention in \eqref{functions-of-a-vector-convention})
\be\label{LC[f]-def}
\vec{L}_{\vec{C}}[\vec{f}](x):=\vec{G}\cdot\log\left(e^{-\vec{a}(x)}+e^{\vec{f}(x)}\right)-\vec{C}\cdot\vec{f}(x) \ .
\ee
With this notation, the TBA equation \eqref{TBAintro} reads
\be\label{TBAshort}
\vec{f}(x) = \left(\Phi_{\vec{C}}\star \vec{L}_{\vec{C}}[\vec{f}]\right)(x)\ .
\ee
	From Proposition~\ref{TBAuniquenessDynkin} we know that 
\eqref{TBAshort} has a unique bounded continuous solution for one special choice of $\vec{C}$, namely $\vec{C}=\frac{1}{2}\vec{G}$. In this section we will apply the results of section \ref{section-soldefeqn} in order to translate that statement to other choices of $\vec{C}$, as well as to the associated Y-system.
In particular, we will prove Theorems~\ref{main-theorem-Y}, \ref{main-theorem-TBA} and Corollary~\ref{unique-constant-sol} from the introduction.

\subsection{Independence of the choice of $\vec C$}

In this subsection we fix
\be
	\vec{G}\in\mathrm{Mat}(N,\mathbb{R}) \quad , \qquad
	\vec{C}\in\mathrm{Mat}_{<2}(N) \quad , \qquad
	\vec{a}\in BC_-(\mathbb{R},\mathbb{R})^N \ .
\ee
We stress that for the moment, we make no further assumptions on $\vec G$ (as opposed to Theorems~\ref{main-theorem-Y} and \ref{main-theorem-TBA}).

We will later need to apply Proposition~\ref{FuncRelToNLIE} to \eqref{TBAshort}. To this end we now provide a criterion for the components of $\vec{L}_{\vec{C}}[\vec{f}]$ to be H\"older continuous. 

\begin{lemma} 
	Let $\vec{f}\in BC(\mathbb{R},\mathbb{R})^N$.
If the components of $\vec{f}$ and of $e^{-\vec{a}}$ are H\"older continuous, then the components of $\vec{L}_{\vec{C}}[\vec{f}]$ are H\"older continuous.
\end{lemma}

\begin{proof}
It is easy to see that the composition of H\"older continuous functions is again H\"older continuous, as is the sum of bounded H\"older continuous functions.
Therefore, and since $x\mapsto e^x$ is H\"older continuous on any compact subset of $\mathbb{R}$ -- in particular on the images of the bounded functions $f_m$ -- the functions $x\mapsto u_m(x):=e^{-a_m(x)}+e^{f_m(x)}$ are H\"older continuous. The bounds on $a_m$ and $f_m$ ensure that the image of $u_m$ is contained in some interval $[x_0,x_1]$ with $x_0,x_1>0$.
But $x\mapsto \log(x)$ is H\"older continuous on $[x_0,x_1]$, and so the functions $x\mapsto l_m(x):=\log(u_m(x))$ are 
	bounded and
H\"older continuous. From this it follows that the components of $\vec{L}_{\vec{C}}[\vec{f}](x)$,
\be
x \,\longmapsto\, \left[\,\vec{L}_{\vec{C}}[\vec{f}]\,\right]_n (x)
= \sum_{m=1}^N \left(G_{nm}l_m(x)-C_{nm}f_m(x)\right) \ ,
\ee
are H\"older continuous.
\end{proof}

We are careful not to make too strong assumptions on $\vec{a}$ here, namely we do not require the components of $\vec{a}$ to be H\"older continuous. For instance, the relevant example $a_m(x) \sim e^{\gamma x/s}$ from \eqref{asymptoticexample} is only locally H\"older continuous.
Meanwhile, the H\"older condition on $\vec{f}$ is, in fact, obtained from the TBA equation for free:

\begin{lemma} \label{alwayshoelder} Suppose $\vec{f}\in BC(\mathbb{R},\mathbb{R})^N$ is a solution of the TBA equation \eqref{TBAshort}. Then the components of $\vec{f}$ are 
Lipschitz continuous.
\end{lemma}

\begin{proof}
Using Lemma~\ref{Phiproperties}\,\ref{matrixelements} one quickly verifies that one can 
write the components of $\vec{f}$ as real linear combinations
of the form
\be
f_n(x)=\sum_{d,m} c_{d,m}^{(n)}F_d\big[\left[\,\vec{L}_{\vec{C}}[\vec{f}]\,\right]_m \big](x) \ ,
\ee 
where $F_d[-]$ is the convolution functional defined in \eqref{Fd-convolution-def}, see also \eqref{phi*g-via-Fd}.
Lemma~\ref{DerivConvBounded} then implies that the derivatives $\frac{d}{dx}f_n(x)$ are bounded. 
	This shows the claim (cf.\ Lemma~\ref{pLipschitz}).
\end{proof}

Now we are in the position to apply Proposition~\ref{FuncRelToNLIE} to \eqref{TBAshort}.
	The $\vec{C}$-independence will boil down to the following simple observation on the functional equation \eqref{ffuncrel}: the $\vec{C}$-dependence on the left and right hand side of
\be\label{ffC=LC-v1}
\vec{f}(x+is)+\vec{f}(x-is) -\vec{C}\cdot\vec{f}(x) = \vec{L}_{\vec{C}}[\vec{f}](x) 
\ee
simply cancels, see \eqref{LC[f]-def}. Thus \eqref{ffC=LC-v1} is in particular equivalent to
\be\label{ffC=LC-v2}
\vec{f}(x+is)+\vec{f}(x-is) = \vec{L}_{\vec{0}}[\vec{f}](x) 
\ee

\begin{proposition}\label{uniquenessFsystem}
Suppose that the components of $e^{-\vec{a}}$ are H\"older 
continuous
	and that 
there exists $\vec{C}\in\mathrm{Mat}_{<2}(N)$, such that the TBA equation
\eqref{TBAshort}
has a unique solution $\vec{f}_\star$ in $BC(\mathbb{R},\mathbb{R})^N$. 
Then:
\begin{enumerate}[label=\roman*)]
\item \label{TBAtoF} 
$\vec{f}_\star$ is real analytic and can be continued to a function in $\mathcal{BA}(\mathbb{S}_s)^N$, which we also denote by $\vec{f}_\star$. It is the unique solution to the functional equation
\be\label{funeq}
\vec{f}(x+is)+\vec{f}(x-is) = \vec{L}_{\vec 0}[\vec{f}](x) \quad , \qquad x\in\mathbb{R} \ ,
\ee
in the space $\mathcal{BA}(\mathbb{S}_s)^N$
	which also satisfies
	$\vec{f}(\mathbb{R})\subset\mathbb{R}^N$.
\item \label{indepC} For any $\vec{C}'\in\mathrm{Mat}_{<2}(N)$, 
$\vec{f}_\star$ 
is the unique solution to the TBA equation
\be\label{Cprimetba}
\vec{f}(x) = \left(\Phi_{\vec{C}'}\star \vec{L}_{\vec{C}'}[\vec{f}]\right)(x)
\ee
in the space $BC(\mathbb{R},\mathbb{R})^N$.
\end{enumerate}
\end{proposition}

\begin{proof}
\ref{TBAtoF} By definition, the components of $\vec{g}_\star(x):=\vec{L}_{\vec{C}}[\vec{f}_\star](x)$ are bounded real functions. Moreover, Lemma~\ref{alwayshoelder} ensures that they are Lipschitz continuous, so in particular H\"older continuous. It follows from 
	direction $\ref{ffg-inteq} \Rightarrow \ref{ffg-funrel}$ in
Proposition~\ref{FuncRelToNLIE} that $\vec{f}_\star\in\mathcal{BA}(\mathbb{S}_s)^N$, and that $\vec{f}_\star$ satisfies the functional relation
\be\label{auxfuneq}
\vec{f}_\star(x+is)+\vec{f}_\star(x-is) - \vec{C}\cdot \vec{f}_\star(x) = \vec{L}_{\vec{C}}[\vec{f}_\star](x)
\ee
for all $x\in\mathbb{R}$,
	which, as we just said, is equivalent to \eqref{funeq}.

Now suppose there is another solution $\vec{f}'_\star\in\mathcal{BA}(\mathbb{S}_s)^N$ to
	\eqref{funeq}, or, equivalently, \eqref{auxfuneq}, which satisfies $\vec{f}'_\star(\mathbb{R})\subset\mathbb{R}^N$.
	By direction $\ref{ffg-funrel} \Rightarrow \ref{ffg-inteq}$
of Proposition~\ref{FuncRelToNLIE}, the restriction $\vec{f}'_\star\big|_\mathbb{R}$ is also a solution to the TBA equation \eqref{TBAshort}. 
	By our uniqueness assumption, we must have $\vec{f}'_\star\big|_\mathbb{R} = \vec{f}_{\star}\big|_\mathbb{R}$, and by uniqueness of the analytic continuation also $\vec{f}'_\star = \vec{f}_{\star}$ on $\overline{\mathbb{S}}_s$.

\medskip

\noindent
\ref{indepC} For any choice of $\vec{C}'\in\mathrm{Mat}_{<2}(N)$,
\eqref{auxfuneq} can be rewritten as 
\be \label{ffC=LC-v3}
	\vec{f}_\star(x+is)+\vec{f}_\star(x-is) - \vec{C}'\cdot \vec{f}_\star(x) = \vec{L}_{\vec C'}[\vec{f}_\star](x) \ .
\ee
Direction $\ref{ffg-funrel} \Rightarrow \ref{ffg-inteq}$ of Proposition~\ref{FuncRelToNLIE} shows that $\vec{f}_\star$ satisfies \eqref{Cprimetba}.
	Suppose $\vec{f}'_\star \in BC(\mathbb{R},\mathbb{R})^N$ is another solution to \eqref{Cprimetba}. Then by Lemma~\ref{alwayshoelder}, $\vec{f}'_\star$ is H\"older continuous, and by direction $\ref{ffg-inteq} \Rightarrow \ref{ffg-funrel}$ of Proposition~\ref{FuncRelToNLIE} it satisfies \eqref{ffC=LC-v3}. But \eqref{ffC=LC-v3} is equivalent to \eqref{funeq}, whose solution is unique and equal to $\vec f_\star$ by part \ref{TBAtoF}.
\end{proof}

\subsection{Proofs of Theorems \ref{main-theorem-Y}, \ref{main-theorem-TBA} and of Corollary~\ref{unique-constant-sol}}

We now turn to the proof of the main results of this paper. We start with
	part of 
Theorem~\ref{main-theorem-TBA}, as this is used in the proof of Theorem \ref{main-theorem-Y}. 
	Then we show Theorem~\ref{main-theorem-Y} and subsequently the missing part of Theorem~\ref{main-theorem-TBA}.
Finally, we give the proof of Corollary~\ref{unique-constant-sol}.

\subsubsection*{Proof of Theorem \ref{main-theorem-TBA}, less uniqueness in part \ref{mainTBAsolvesY}}

\noindent 
{\em Part \ref{mainTBAunique}:} First note that the assumptions in Proposition~\ref{TBAuniquenessDynkin} are satisfied. Thus there exists a unique solution $\vec f_\star$ to \eqref{TBAintro} in $BC(\mathbb{R},\mathbb{R})^N$ for the specific choice $\vec G = \tfrac12 \vec C$. 
Therefore, the conditions of Proposition~\ref{uniquenessFsystem} are satisfied and part \ref{indepC} of that proposition establishes 
	existence and uniqueness of a solution $\vec f_\star \in BC(\mathbb{R},\mathbb{R})^N$ to \eqref{TBAintro} for any choice of $\vec C \in \mathrm{Mat}_{<2}(N)$, as well as
$\vec C$-independence
of that solution.

\medskip

\noindent 
{\em Part \ref{mainTBAsolvesY} (without uniqueness):} By Proposition~\ref{uniquenessFsystem}\,\ref{TBAtoF}, $\vec{f}_\star$ can be analytically continued to a function in $\mathcal{BA}(\mathbb{S}_s)^N$, which we will also denote by $\vec{f}_\star$.
We define
\be\label{Y-via-f*-proof}
	\vec{Y}(z):=\exp\!\big(\,\vec{a}(z)+\vec{f}_\star(z)\,\big) \ .
\ee
Let us denote the components of $\vec{f}_\star$ and $\vec{Y}$ by 
$f_{\star,n}$ and $Y_{n}$, respectively. Note that $Y_n \in \mathcal{A}(\mathbb{S}_s)$.

It is immediate that the $Y_n$ satisfy properties \ref{Yproperties:real}--\ref{Yproperties:asymptotics} of Theorem~\ref{main-theorem-Y} (for property \ref{Yproperties:real} note that $f_{\star,n}$ and $a_n$ are real-valued on the real axis).
Furthermore, 
as a consequence of the functional relations \eqref{linear-eqn-for-asympt} and \eqref{funeq} for $\vec{a}$ and $\vec{f}_\star$ respectively, the $Y_n$ solve \eqref{Y}:
\begin{align}
 Y_{n}(x+is)Y_{n}(x-is) &= e^{a_n(x+is)+a_n(x-is)}e^{f_{\star,n}(x+is)+f_{\star,n}(x-is)}\nonumber\\
 &= e^{\sum_m G_{nm}a_m(x)}e^{\sum_m G_{nm}\log\left(e^{-a_m(x)}+e^{f_{\star,m}(x)}\right)} \nonumber\\
 &= \prod_m e^{G_{nm}a_m(x)}\left(e^{-a_m(x)}+e^{f_{\star,m}(x)}\right)^{G_{nm}} \nonumber\\
 &= \prod_m \left(1 + Y_{m}(x) \right)^{G_{nm}} \ .
\end{align}

This proves that $\vec Y$ is a solution to \eqref{Y} which lies in $\mathcal{A}(\mathbb{S}_s)^N$ and satisfies properties \ref{Yproperties:real}--\ref{Yproperties:asymptotics} in Theorem~\ref{main-theorem-Y}. It remains to show that $\vec Y$ is the unique such solution.
This will be done below as an immediate consequence of Theorem~\ref{main-theorem-Y}, whose proof we turn to now.

\subsubsection*{Proof of Theorem \ref{main-theorem-Y}}

Existence of a $\vec Y \in \mathcal{A}(\mathbb{S}_s)^N$ which solves  \eqref{Y} and satisfies properties \ref{Yproperties:real}--\ref{Yproperties:asymptotics} has just been proven above. The solution $\vec Y$ is given via \eqref{Y-via-f*-proof} in terms of the valid asymptotics $\vec a$ and the unique solution $\vec f_\star$ to \eqref{TBAintro} obtained in Theorem~\ref{main-theorem-TBA}\,\ref{mainTBAunique}.
It remains to show uniqueness of $\vec Y$.

Suppose there is another function $\vec{Y}'\in\mathcal{A}(\mathbb{S}_s)^N$ with the properties \ref{Yproperties:real}--\ref{Yproperties:asymptotics}. Since $\mathbb{S}_s$ is a simply connected domain and the components $Y_{n}'(z)$ 
have no roots in $\overline{\mathbb{S}}_s$ (property \ref{Yproperties:roots}), there exists a function $\vec{h}\in\mathcal{A}(\mathbb{S}_s)^N$, such that $\vec{Y}'(z)=\exp(\vec{h}(z))$ for all $z\in\overline{\mathbb{S}}_s$. 
In fact, property \ref{Yproperties:real} (real\ \& positive) allows one to choose this function in such a way that $\vec{h}(\mathbb{R})\subseteq\mathbb{R}^N$. 
Consequently, the function $\vec{f}'_\star(z):=\vec{h}(z)-\vec{a}(z)$ is in $\mathcal{A}(\mathbb{S}_s)^N$ and satisfies $\vec{f}'_\star(\mathbb{R})\subseteq\mathbb{R}^N$. 
Due to property \ref{Yproperties:asymptotics} (asymptotics),
$\vec{f}'_\star\in\mathcal{BA}(\mathbb{S}_s)^N$. 
As a consequence of the Y-system \eqref{Y} and the functional relation \eqref{linear-eqn-for-asympt} for the asymptotics, we have
\begin{align}
e^{f'_{\star,n}(x+is)+f'_{\star,n}(x-is)} &= e^{-a_n(x+is)-a_n(x-is)}\,Y'_n(x+is)Y'_n(x-is) \nonumber \\
&= \prod_m e^{-G_{nm}a_m(x)} \left(1 + Y'_m(x) \right)^{G_{nm}} \nonumber \\
&=\prod_m  \left(e^{-a_m(x)} + e^{f'_{\star,m}(x)} \right)^{G_{nm}} \ .
\end{align}
Hence, due to continuity of $\vec{f}'_\star(x+is)+\vec{f}'_\star(x-is)- \vec{L}_{\vec 0}[\vec{f}'_\star](x)$
there exists $\vec{v}\in\mathbb{Z}^N$, such that
\be\label{expfsystem-aux}
\vec{f}'_\star(x+is)+\vec{f}'_\star(x-is)+2\pi i \vec{v} = \vec{L}_{\vec 0}[\vec{f}'_\star](x) \ .
\ee
But since $\vec{f}'_\star(\mathbb{R})\subseteq\mathbb{R}^N$, the Schwarz reflection principle $\vec{f}'_\star(x-is) = \overline{\vec{f}'_\star(x+is) }$ allows \eqref{expfsystem-aux} to be rewritten as
\be 
	2\,\mathrm{Re}\,\vec{f}'_\star(x+is)+2\pi i \vec{v} = \vec{L}_{\vec 0}[\vec{f}'_\star](x) \ .
\ee
Since the right hand side is real, it follows that $\vec{v}=0$. Thus $\vec{f}'_\star$ solves \eqref{funeq}. 

As in the previous proof, by Proposition~\ref{TBAuniquenessDynkin} we can apply Proposition~\ref{uniquenessFsystem}. Part \ref{TBAtoF} of the latter proposition states that 
$\vec{f}'_\star$ is the unique solution to \eqref{funeq} in $\mathcal{BA}(\mathbb{S}_s)^N$ which maps $\mathbb{R}$ to $\mathbb{R}^N$.
Part \ref{indepC} states that $\vec{f}'_\star$ solves the TBA equation \eqref{Cprimetba} for any choice of $\vec C'$, so in particular it solves \eqref{TBAintro}.
But from Theorem~\ref{main-theorem-TBA}\,\ref{mainTBAunique} we know that $\vec f_\star$ is the unique solution to \eqref{TBAintro} in $BC(\mathbb{R},\mathbb{R})^N$, and hence $\vec{f}'_\star=\vec{f}_\star$ on $\mathbb{R}$ (and therefore, by uniqueness of the analytic continuation,
 also on $\overline{\mathbb{S}}_s$).

This completes the proof of Theorem~\ref{main-theorem-Y}.

\subsubsection*{Proof of Theorem \ref{main-theorem-TBA}, uniqueness in part \ref{mainTBAsolvesY}}

This is now immediate from Theorem~\ref{main-theorem-Y}, as the solution to the Y-system in $\mathcal{A}(\mathbb{S}_s)^N$ satisfying \ref{Yproperties:real}--\ref{Yproperties:asymptotics} is unique.

This completes the proof of Theorem~\ref{main-theorem-TBA}.

\subsubsection*{Proof of Corollary~\ref{unique-constant-sol}}

We only need to show that for $\vec a=0$, the unique solution $\vec f_\star$ 
from Theorem~\ref{main-theorem-TBA} is constant. By part 2 of that theorem, $\vec Y$ is then constant, too. 

By Theorem~\ref{main-theorem-TBA}\,\ref{mainTBAunique}, the solution $\vec f_\star$ is independent of $\vec C$, and in particular is equal to the unique solution found in Proposition~\ref{TBAuniquenessDynkin} for $\vec C = \tfrac12 \vec G$. 
In the proof of Proposition~\ref{TBAuniquenessDynkin} in Section~\ref{uniqunessproofTBA} it was verified that the integral operator
$I : BC(\mathbb{R},\mathbb{R})^N \to BC(\mathbb{R},\mathbb{R})^N$ from \eqref{contractingintegralforg} is a contraction. Explicitly,
\be \label{contTBAoperator}
	I_i[\vec g] := 
\sum_{j=1}^N
\int_{-\infty}^{\infty} \frac{1}{w_i}\phi_{ij}(x-y)\sum_{k=1}^NG_{jk}\left[\log\left(e^{-a_k(y)}+e^{w_k g_k(y)}\right)-\tfrac12 w_k g_k(y)\right] dy
\ee
where $\phi_{ij}$ are the entries of $\Phi_{\frac12 \vec G}$ and $\vec w$ is the Perron-Frobenius eigenvector of $\vec{G}$.
The relation between $\vec g$ and the functions $\vec f$ in the TBA equation is $f_i(x) = w_i g_i(x)$. 

By assumption we have $\vec a = 0$. Clearly, if also $\vec g$ is constant, so is $I[\vec g]$. 
This shows that the operator $I$ preserves the space of constant functions (for $\vec a=0$). Hence the unique fixed point $\vec g_\star$ of $I$ (and hence the unique solution $\vec f_\star$ to \eqref{TBAintro}) must be a constant function. 

This completes the proof of Corollary~\ref{unique-constant-sol}.

\section{Discussion and outlook}\label{section-outlook}

In this section, we will make additional
comments on our results and discuss possible further investigations, in particular with reference to the physical background.

\medskip

First of all, let us mention that the existence of a unique solution to Y-systems or TBA equations, even if the latter arise from a physical context, is by no means clear:
\begin{itemize}
\item \emph{Existence}: as already mentioned earlier, the constant Y-system \eqref{Yconst} has no non-negative solution if the spectral radius of $\vec G$ is larger or equal than 2 \cite{Tateo:DynkinTBAs}. 
It seems likely that this generalises to solutions with asymptotics $\vec{a} = \vec{w}e^{\pm\gamma x/s}$, cf.\ \eqref{asymptoticexample}.
\item \emph{Uniqueness}: when relaxing the reality condition $\vec{Y}(\mathbb{R})\subset\mathbb{R}$, uniqueness generally fails to hold (see \cite[Fig.\,1]{DoreyTateo:YL} for the case $N=1$). 
Moreover, there exist Y-systems of some more general form for which stability investigations show that the associated TBA equation for constant functions is not a contraction,
and in fact may display chaotic behaviour upon iteration \cite{CastroAlvaredoFring}.\footnote{Note that our results do not imply that the TBA integral operators are contracting, except in the specific case in Section~\ref{uniqunessproofTBA} to which the Banach Theorem is applied. However, our bound on the contraction constant $\kappa$ is certainly not optimal (see Remark~\ref{kappa-is-not-optimal}) and the TBA integral operator may well be contracting even if our bound yields $\kappa>1$. The results of $\cite{CastroAlvaredoFring}$ show that sometimes it is not contracting.}
\end{itemize}

\medskip

We will conclude this work with a brief outline of possible future investigations. Our results already cover a significant class of integrable models \cite{Zamo:ADE,KlassenMelzer91}. In physical terms, the state corresponding to the unique solution with no zeros in $\overline{\mathbb{S}}_s$ is the
	ground state.
Two main directions in which to extend our results are to a) consider excited states, and b) treat situations associated with more general models. We will briefly comment on both of these.

\subsubsection*{a) Including excited states}

To include excited states one has to allow for $Y$-functions which
have roots in $\overline{\mathbb{S}}_s$. 
In this case, 
the TBA equation involves an additional term which depends on the positions of these roots \cite{KluemperPearce,DoreyTateo:YL,BLZ4}.
It seems possible that existence and
uniqueness of a solution to the TBA equation for a generic set of root positions can be established with similar techniques as in the present paper.
However, to produce a solution to the Y-system with a sufficiently far analytic continuation, one needs to impose additional constraints. 
It would be very interesting to understand if some general statements about solutions to these constraints can be made.

In examples, these solutions are parametrised by a discrete set of ``quantum numbers''.
In the asymptotic limit $r\to \infty$ of Y-systems with $\vec{a}(x) = r \cosh(\gamma x/s) \, \vec{w}$ (relativistic scattering theories in volume $r$) these quantum numbers are expected to coincide with the Bethe-Yang quantum numbers \cite{YangYang69} which parametrise solutions of the Bethe ansatz equations. In the limit $r\to 0$, quantisation conditions in terms of Virasoro states have been conjectured for the Yang-Lee model ($N=1,\vec{G}=1$), see \cite{BajnokDeebPearce}. A related $N=1$ model, albeit with a slightly deformed Y-system (see below), is the Sinh-Gordon model for which a conjecture on the classification of states (for all $r$) in terms of solutions to the Y-system
	has been given in 
\cite{Teschner}.

\subsubsection*{b) More general Y-systems}

It is fair to ask if the approach presented in this paper is flexible enough to make contact with a larger number of physical models. This would require us to consider different conditions on $\vec{G}$ 
as well as
Y-systems or TBA equations of a more general form.

\medskip

There are several generalisations which still fit the form \eqref{Y}:
\begin{itemize}
\item 
As mentioned in the end of Section~\ref{N=1-example-section}, 
	with the method used in \cite{FringKorffSchulz}
it may be possible to treat cases of slightly more general $\vec{G}$ giving rise to $\kappa=1$, where our results
	no longer apply.
\item Giving up the reality requirement in a controlled way would
for example allow for a treatment of Y-systems with chemical potentials, where a constant imaginary vector is added to $\vec{a}$
 (see e.g.\ \cite{KlassenMelzer91,Fendley92}).
\end{itemize}

There are also more general forms of \eqref{Y}, which would be of interest. For example:
\begin{itemize}
\item Y-systems with a second shift parameter $t$ such as 
\be \label{affineToda}
Y_n(x+is)Y_n(x-is)=\frac{\left(1+Y_n(x+it)\right)\left(1+Y_n(x-it)\right)}{\prod_{m=1}^N \left(1+\frac{1}{Y_m(x)}\right)^{H_{nm}}} \,
\ee
where $\vec{H}$ is the adjacency matrix of a finite Dynkin diagram. This specific type of Y-system appears in the context of Affine Toda Field Theories \cite{Martins,FringKorffSchulz} whose simplest representative, the Sinh-Gordon model, has $\vec{H}=0$
	and bears some resemblance to \eqref{Y}.
\item The case of two simple Lie algebras giving rise to a Y-system of the form
\be
Y_{n,m}(x+is)Y_{n,m}(x-is)=\frac{\prod_{k=1} \left(1+Y_{k,m}(x)\right)^{G_{nk}}}{\prod_{l=1} \left(1+\frac{1}{Y_{n,l}(x)}\right)^{H_{ml}}}.
\ee
via their Dynkin diagrams $\vec{G}$ and $\vec{H}$.
In applications, often many of the $Y_{n,m}$ are required to be trivial. One example, albeit with an infinite number of Y-functions, is the famous ``T-hook'' of the AdS/CFT Y-system 
	(see e.g.\ \cite{Bajnok:TBAreview} for more details and references). 
The physically relevant solutions in this case have, however, rather complicated analytical properties involving also branch cuts.
\end{itemize}

\appendix

\section{Appendix}

\subsection{Fourier transformation of $\cosh(x)^{-m}$}
\label{Fourier_coshm}

\begin{lemma} Let $m\in\mathbb{Z}_{\geq 1}$, and $f_m(x)=\cosh(x)^{-m}$. The Fourier transformation of $f_m(x)$ is given by
\be\label{fouriercosh-m}
\hat{f}_m(k):=\int_{-\infty}^\infty e^{-ikx}f_m(x) \; dx = \frac{\pi}{(m-1)!}\left(\prod_{\substack{l=m-2\\ \mathrm{step\, -2}}}^1 (k^2 + l^2)\right)\cdot \begin{cases} \frac{1}{\cosh\left(\frac{\pi}{2}k\right)} & \mathrm{if} \; m \;\mathrm{odd} \\ \frac{k}{\sinh\left(\frac{\pi}{2}k\right)} &\mathrm{if} \; m \;\mathrm{even} \end{cases}\, .
\ee
\end{lemma}

\begin{proof}
It is convenient to first get rid of the infinite number of poles of the integrand. This is achieved by the variable transformation $r=e^{2x}$, which results in
\be
\hat{f}_m(k) = 2^{m-1}\int_0^\infty \frac{\left(r^{1/2}\right)^{m-2-ik}}{(1+r)^m}\;dr \ .
\ee
This transformation comes at the expense of single-valuedness: the new integrand has only one simple pole at $r=-1$, but also a branch cut connecting the origin and infinity via, say, the negative imaginary axis. Since the integrand decays fast enough for $|r|\to\infty$, it is possible to revolve the integration contour once around the origin like the big hand of a watch. If we do this \textsl{counter-clockwise},  the whole integral picks up a residue, and the square-root acquires a monodromy of $-1$, which changes the numerator to
\be
\left(-r^{1/2}\right)^{m-2-ik} = (-1)^m e^{k\pi} \left(r^{1/2}\right)^{m-2-ik} \ .
\ee
This contour manipulation yields the equation
\begin{align}
\int_0^\infty \frac{\left(r^{1/2}\right)^{m-2-ik}}{(1+r)^m}\;dr = 2\pi i &\;\mathrm{Res}_{z=-1}\left(\frac{\left(z^{1/2}\right)^{m-2-ik}}{(1+z)^m}\right) \nonumber \\ &+ (-1)^me^{k\pi} \int_0^\infty \frac{\left(r^{1/2}\right)^{m-2-ik}}{(1+r)^m}\;dr \ .
\end{align}
This can now be solved for the integral:
\be\label{ds9w2}
\hat{f}_m(k) = \frac{1}{1-(-1)^me^{k\pi}} 2^m i\pi \; \mathrm{Res}_{z=-1}\left(\frac{\left(z^{1/2}\right)^{m-2-ik}}{(1+z)^m}\right) \ .
\ee
The residue can be computed as the coefficient of the Laurent-expansion
\be
\left(-z^{1/2}\right)^a = \sum_{n=0}^\infty (-1)^n i^a (1+z)^n {{\frac{a}{2}}\choose{n}} 
\ee
for $n=m-1$, namely
\be
\mathrm{Res}_{z=-1}\left(\frac{\left(z^{1/2}\right)^{m-2-ik}}{(1+z)^m}\right) = e^{\frac{\pi}{2}(k-im)}{{\frac{1}{2}(m-2-ik)}\choose{m-1}} \ .
\ee
Plugging this into \eqref{ds9w2} gives
\be \label{ash8d}
\hat{f}_m(k) = 2^{m-1}\pi e^{i\frac{3}{2}\pi(1-m)}{{\frac12(m-2-ik)}\choose{m-1}} \cdot \begin{cases} \frac{1}{\cosh\left(\frac{\pi}{2}k\right)} & \mathrm{if} \; m \;\mathrm{odd} \\ \frac{k}{\sinh\left(\frac{\pi}{2}k\right)} &\mathrm{if} \; m \;\mathrm{even} \end{cases}\ .
\ee
Finally, the following identity is straight-forward to prove by induction $m\rightarrow m+2$:
\be\label{ws87w}
\prod_{j=1}^{m-1}\left(m-2j-ik\right) = e^{i\frac{3}{2}\pi(m-1)} \left(\prod_{\substack{l=m-2\\ \mathrm{step\,  -2}}}^1 (k^2 + l^2)\right)\cdot \begin{cases} 1 & \mathrm{if} \; m \;\mathrm{odd} \\ k &\mathrm{if} \; m \;\mathrm{even} \end{cases}
\ee
Substituting \eqref{ws87w} into \eqref{ash8d} results in \eqref{fouriercosh-m}.
\end{proof}

\subsection{Sokhotski integrals} \label{AppendixSokhotski}

The following proposition is the key ingredient in the proof of Proposition~\ref{SokhotskiConvolutionMain} in Appendix~\ref{AppendixPropProof}. The proposition and its proof are adapted from \cite[Ch.\,1,\,\S4]{Gakhov}, where a version of this theorem with contours of general shape but finite length is treated, and where the functions $\varphi(z,t)$ below are constant in $z$.

\begin{proposition}\label{Sokhotski}
Let $D\subseteq\mathbb{C}$ be a complex domain such that $\mathbb{S}_a\subseteq D$ 
for some $a>0$. Let $\varphi: D\times\mathbb{R} \rightarrow \mathbb{C}$ be a function with the following properties:
\begin{enumerate}
\item \label{Sokhotski:1} (Analyticity) For every $t_0\in\mathbb{R}$, the function $z\mapsto\varphi(z,t_0)$ is analytic in $D$.
\item \label{Sokhotski:2} (H\"older-continuity) There exist $0<\alpha\leq 1$ and $C>0$, such that for every $z_0\in D$ the function $t\mapsto\varphi(z_0,t)$ is $\alpha$-H\"older continuous with H\"older constant $C$.
\item \label{Sokhotski:3} (Decay) There exist $\mu>0$ and $T>0$, such that  $|\varphi(z,t)|\leq|t-z|^{-\mu}$ for all $z\in D$ and $t\in\mathbb{R}$ with $|t-z| \geq T$.
\item \label{Sokhotski:4} (Local majorisation) For every $z_0\in D\setminus \mathbb{R}$ there exist a neighbourhood $U\subseteq D\setminus\mathbb{R}$ and a function $M\in L_1(\mathbb{R})$, such that $|\varphi(z,t)|\leq|z-t|M(t)$ for all $z\in U$.
\item \label{Sokhotski:5} (Uniform convergence) 
The convergence $\varphi(x\pm iy,t)\xrightarrow{y\searrow 0}\varphi(x,t)$
is uniform in
$(x,t) \in \mathbb{R}^2$.
\item \label{Sokhotski:6} (Boundedness) 
$\sup_{(x,t)\in\mathbb{R}^2}|\varphi(x,t)|<\infty$~.
\end{enumerate}
Then the function 
\be\label{Sokh-F(z)-defn}
F(z) = \int_{-\infty}^\infty \frac{\varphi(z,t)}{t-z}\;dt
\ee
is analytic in $D\setminus\mathbb{R}$, and there exist limiting functions $F^\pm :\mathbb{R}\rightarrow\mathbb{C}$ such that 
\be\label{eqn:Sokhotski-uniform}
F(x\pm iy)\rightarrow F^\pm(x)
\ee
uniformly as $y\searrow 0$. The functions $F^\pm(x)$ 
are bounded 
 and satisfy
\be \label{eqn:Sokhotski} F^+(x)-F^-(x) = 2i\pi\varphi(x,x)
	\qquad \text{for all} \quad x \in \mathbb{R} \ .
\ee
\end{proposition}

\begin{proof} ~\\
$\bullet$
{\em $F(z)$ is analytic on $D \setminus \mathbb{R}$:}
Conditions \ref{Sokhotski:2} and \ref{Sokhotski:3} ensure that the integrand in \eqref{Sokh-F(z)-defn} is always in $L_1(\mathbb{R})$. Thus, $F(z)$ is well-defined. Analyticity of $F(z)$ in $D\setminus\mathbb{R}$ follows directly from lemma \ref{analyticity} together with condition \ref{Sokhotski:4}.

\medskip

\noindent
$\bullet$
{\em The auxiliary function $\psi$:}
Below we make frequent use of the following simple integral. Let $x,y,\eta,L \in \mathbb{R}$, $\eta \ge 0$, $\eta \pm x \le L$, and suppose that $y \neq 0$ in case $\eta=0$.
Denote $B_\eta(x)=(x-\eta,x+\eta)$.
 Then
\be\label{adi89}
\int_{[-L,L]\setminus B_\eta(x)}\frac{1}{t-x-iy} \;dt ~=~
\log \frac{L-x-iy}{L+x+iy} ~+~ 
\begin{cases}
\log \frac{\eta+iy}{\eta-iy} &; \eta>0 \\
i \pi \, \mathrm{sgn}(y) &; \eta = 0 
\end{cases} \quad .
\ee
Here, the branch cut of the logarithm is placed along the negative real axis.

We now investigate the $y\searrow 0$ limit of $F(x \pm iy)$. To do so, we split $F(z)$ into two integrals by adding and subtracting a term in the integrand. Namely, for $z \in D\setminus\mathbb{R}$ we have
\be\label{sd68s}
F(z)
=\lim_{L\rightarrow\infty}\int_{-L}^L \frac{\varphi(z,t)-\varphi(z,\mathrm{Re}(z))}{t-z}\;dt +\varphi(z,\mathrm{Re}(z)) \lim_{L\rightarrow\infty}\int_{-L}^L\frac{1}{t-z} \;dt \ .
\ee
The improper integral 
\be\label{psi-improper-def}
\psi(z):=\lim_{L\rightarrow\infty}\int_{-L}^L \frac{\varphi(z,t)-\varphi(z,\mathrm{Re}(z))}{t-z}\;dt
\ee
 exists because by \eqref{adi89} the limit in the second summand of \eqref{sd68s} exists. We obtain, for $z \in D\setminus\mathbb{R}$,
	\be\label{82gfs}
F(z)=\psi(z) + i\pi \, \varphi(z,\mathrm{Re}(z)) \, \mathrm{sgn}(\mathrm{Im}(z))\ .
\ee
Since both $F(z)$ and $\varphi(z,\mathrm{Re}(z))$ are continuous in $D\setminus \mathbb{R}$, $\psi(z)$ is also continuous in $D\setminus \mathbb{R}$. 

We now claim that the integral and limit defining $\psi(z)$ in \eqref{psi-improper-def} also exist for $z \in \Rb$.
To see this, first note that due to the H\"older condition (condition \ref{Sokhotski:2}) 
\be
\left|\frac{\varphi(x,t)-\varphi(x,x)}{t-x}\right|\leq \frac{C}{|t-x|^{1-\alpha}} \hspace{1cm}\forall x,t\in\mathbb{R} \ .
\ee
Hence,  the integral 
\be\label{smd834ha}
\int_{x-1}^{x+1} \frac{\varphi(x,t)-\varphi(x,x)}{t-x}\;dt
\ee
exists. On the other hand, by \eqref{adi89} and for $|x|+1<L$,
\be\label{smd834hb}
\int_{[-L,L]\setminus B_1(x)} \frac{\varphi(x,t)-\varphi(x,x)}{t-x}\;dt 
= \int_{[-L,L]\setminus B_1(x)} \frac{\varphi(x,t)}{t-x}\;dt -\varphi(x,x)
\log \frac{L-x}{L+x} \ .
\ee 
The integral in the first summand has a well-defined $L \to \infty$ limit by condition \ref{Sokhotski:3}
Adding \eqref{smd834ha} and \eqref{smd834hb} shows that the limit and integral in \eqref{psi-improper-def} exist also for $z \in \mathbb{\Rb}$, so that altogether $\psi$ is defined on all of $D$.

\medskip

\noindent
$\bullet$
{\em Uniform convergence of $\psi$:}
Next we study the continuity properties of $\psi$ on $D$.
Let us restrict $\psi(z)$ to lines parallel to the real axis. Namely, for $|y|<a$ we define $\psi^{[y]}:\mathbb{R}\rightarrow\mathbb{C}$ by 
$\psi^{[y]}(x):=\psi(x+iy)$. 
We will show that $\psi^{[y]}(x)$ converges to ${\psi^{[0]}(x)}$ uniformly as $y\rightarrow 0$ (from both sides). To do so, it is convenient to define the family of functions 
\be
\Delta^{[y]}(x,t):=\varphi(x+iy,t)-\varphi(x,t)\ .
\ee
 Due to the uniform convergence condition on $\varphi(x+iy,t)$ (condition \ref{Sokhotski:5}), $\Delta^{[y]}(x,t)$ converges to 0 uniformly in 
	$(x,t) \in \mathbb{R}^2$
 as $y\rightarrow 0$. In particular, for $|y|$ small enough, $\Delta^{[y]}(x,t)$ is bounded. 
 Moreover, $\Delta^{[y]}(x,t)$ inherits the decay property (condition \ref{Sokhotski:3}) 
 from $\varphi(z,t)$. These properties will be used later in the proof.
 
Choose $\eta>0$ such that $\eta<T$ and split the integration over the interval $[-L,L]$ (w.l.o.g. $|x|+\eta<L$) into the interval $B_\eta(x)$
and its complement $[-L,L]\setminus B_\eta(x)$. A straightforward computation yields
\begin{align}
&\psi^{[y]}(x)-\psi^{[0]}(x)
\nonumber\\
&=   iy\int_{B_\eta(x)} \frac{\varphi(x,t)-\varphi(x,x)}{(t-x)(t-x-iy)}\;dt + iy\lim_{L\rightarrow\infty}\int_{[-L,L]\setminus B_\eta(x)} \frac{\varphi(x,t)-\varphi(x,x)}{(t-x)(t-x-iy)}\;dt 
\nonumber\\
&\phantom{=}+ \int_{B_\eta(x)} \frac{\Delta^{[y]}(x,t)-\Delta^{[y]}(x,x)}{t-x-iy} \;dt +  \lim_{L\rightarrow\infty}\int_{[-L,L]\setminus B_\eta(x)} \frac{\Delta^{[y]}(x,t)-\Delta^{[y]}(x,x)}{t-x-iy} \;dt \ .
\end{align}
We will now show that all four integrals and limits exist and at the same time provide estimates for them. For the first three we compute, where (*) refers to the use of $\alpha$-H\"older continuity (condition~\ref{Sokhotski:2}) and (**) to boundedness (condition~\ref{Sokhotski:6}) -- we set $S := \sup_{(x,t)\in\mathbb{R}^2}|\varphi(x,t)|$,
\allowdisplaybreaks
\begin{align}
&\left|iy\int_{B_\eta(x)} \frac{\varphi(x,t)-\varphi(x,x)}{(t-x)(t-x-iy)}\;dt\right| 
~\leq~ 
\int_{B_\eta(x)} \left|\frac{\varphi(x,t)-\varphi(x,x)}{t-x}\right|\frac{|y|}{|t-x-iy|}\;dt 
\nonumber \\
& \hspace{5em}
\overset{(*)}{\leq} ~C \int_{B_\eta(x)} |t-x|^{\alpha-1} \;dt 
~=~ 
2C \int_0^\eta r^{\alpha-1} \;dr 
~=~ 
\frac{2C\eta^\alpha}{\alpha} 
\ ,
\\
&
\left|\int_{B_\eta(x)} \frac{\Delta^{[y]}(x,t)-\Delta^{[y]}(x,x)}{t-x-iy} \;dt\right|
\nonumber \\
& \hspace{5em}
\leq~ \int_{B_\eta(x)} \frac{\left|\varphi(x+iy,t)-\varphi(x+iy,x)\right|+\left|\varphi(x,t)-\varphi(x,x)\right|}{\left|t-x\right|}\;dt \nonumber \\ 
& \hspace{5em}
\overset{(*)}{\leq}~ 2C  \int_{B_\eta(x)} |t-x|^{\alpha-1} \;dt 
~=~ 
4C \int_0^\eta r^{\alpha-1} \;dr = \frac{4C\eta^\alpha}{\alpha} 
\ ,
\\
&
\left| iy\lim_{L\rightarrow\infty}\int_{[-L,L]\setminus B_\eta(x)} \frac{\varphi(x,t)-\varphi(x,x)}{(t-x)(t-x-iy)}\;dt \right| 
~\leq~ |y|\int_{\mathbb{R}\setminus B_\eta(x)} \frac{|\varphi(x,t)|+|\varphi(x,x)|}{\left|t-x\right|^2}\;dt \nonumber
\\
& \hspace{5em}
\overset{(**)}{\leq}~ 4 S |y| \eta^{-1}  \ .
\end{align}
Now let us turn to the fourth integral, which is slightly more involved. With the help of the decay condition on $\Delta^{[y]}(x,t)$
and \eqref{adi89} we can rewrite it as
\begin{align}
&\lim_{L\rightarrow\infty}\int_{[-L,L]\setminus B_\eta(x)}  \frac{\Delta^{[y]}(x,t)-\Delta^{[y]}(x,x)}{t-x-iy} \;dt 
\nonumber \\ 
&\hspace{5em}
= \int_{\mathbb{R}\setminus B_\eta(x)} \frac{\Delta^{[y]}(x,t)}{t-x-iy} \;dt ~-~ \Delta^{[y]}(x,x) \log \frac{\eta+iy}{\eta-iy} \ .
\end{align}
To estimate the integral over $\mathbb{R}\setminus B_\eta(x)$, we split it as follows, for $R > T$,
\be
\int_{\mathbb{R}\setminus B_\eta(x)} \frac{\Delta^{[y]}(x,t)}{t-x-iy} \;dt = \int_{\mathbb{R}\setminus B_R(x)} \frac{\Delta^{[y]}(x,t)}{t-x-iy} \;dt + \int_{B_R(x)\setminus B_\eta(x)} \frac{\Delta^{[y]}(x,t)}{t-x-iy} \;dt \ .
\ee
We now estimate the two integrals separately:
\begin{align}
\left|\int_{\mathbb{R}\setminus B_R(x)} \frac{\Delta^{[y]}(x,t)}{t-x-iy} \;dt\right|
&
\leq 
\int_{\mathbb{R}\setminus B_R(x)} \left|\frac{\Delta^{[y]}(x,t)}{t-x}\right| \;dt \nonumber\\
&\leq \int_{\mathbb{R}\setminus B_R(x)} \frac{1}{|t-x|^{1+\mu}} \;dt = 2\int_R^\infty \frac{1}{r^{1+\mu}} \;dr  = \frac{2}{\mu R^\mu}
\\
\left|\int_{B_R(x)\setminus B_\eta(x)} \frac{\Delta^{[y]}(x,t)}{t-x-iy} \;dt\right|
&
\leq \int_{B_R(x)\setminus B_\eta(x)} \left|\frac{\Delta^{[y]}(x,t)}{t-x}\right| \;dt 
\leq \frac{1}{\eta}\int_{x-R}^{x+R} |\Delta^{[y]}(x,t)|  \;dt 
\nonumber\\
&
\leq \frac{2R}{\eta} \sup_{t\in\mathbb{R}}|\Delta^{[y]}(x,t)| 
\leq \frac{2R}{\eta} \sup_{t,x\in\mathbb{R}}|\Delta^{[y]}(x,t)|
\end{align}
Finally, we remark that with our choice of branch cut for the logarithm,
\be
\left| \log \tfrac{\eta+iy}{\eta-iy} \right|
 = 2 \left| \mathrm{arg}(\eta + iy) \right|
 \le \pi \ .
\ee
Assembling all of the above estimates, we obtain:
\be\label{psi-unif-conv-estimate}
\left|\psi^{[y]}(x)-\psi^{[0]}(x)\right|\leq \frac{6C\eta^\alpha}{\alpha} 
+  \frac{4 S |y|}{\eta} + 
\frac{2}{\mu R^\mu} + \left(\frac{2R}{\eta}+\pi\right) \sup_{x,t\in\mathbb{R}}|\Delta^{[y]}(x,t)|
\ee

To establish uniform convergence, we need to show that for each $\eps>0$ there exists a $\delta>0$ such that for all $|y|<\delta$ and all $x \in \Rb$ we have $|\psi^{[y]}(x)-\psi^{[0]}(x)| \leq\eps$. To find $\delta$, we choose $\eta$ and $R$ in the above estimate appropriately. 

Choose $\eta$ such that the first term in \eqref{psi-unif-conv-estimate} equals $\varepsilon/4$: 
\be
\eta=\left(\frac{\alpha\varepsilon}{24C}\right)^{\frac{1}{\alpha}}
\ .
\ee
The second term is smaller than $\varepsilon/4$ provided $|y|<\delta_1$, where 
\be
\delta_1= \frac{\eta \, \eps}{16 \, S} = \frac{1}{16 S} \left(\frac{\alpha}{24C}\right)^{\frac{1}{\alpha}}
\eps^{1+\frac1\alpha}
\ .
\ee
 The third term is smaller than $\varepsilon/4$ if we set 
\be
R=\left(\frac{8}{\mu\varepsilon}\right)^{\frac{1}{\mu}}.
\ee
 Finally, we remember that  $\Delta^{[y]}(x,t)\xrightarrow{y\to 0}0$ uniformly in $x$ and $t$. Hence, there exists a $\delta_2>0$ such that for all $|y|< \delta_2$,
\be
\sup_{x,t\in\mathbb{R}}|\Delta^{[y]}(x,t)| \leq \frac{\varepsilon}{4}\left(\frac{2R}{\eta}+\pi\right)^{-1} = \frac{\varepsilon}{4}\left(2\left(\frac{24C}{\alpha\varepsilon}\right)^{\frac{1}{\alpha}}\left(\frac{8}{\mu\varepsilon}\right)^{\frac{1}{\mu}}+\pi\right)^{-1} \ .
\ee
This makes the last term smaller than $\varepsilon/4$.
Setting $\delta :=  \min(\delta_1,\delta_2)$, this proves uniform convergence $\psi^{[y]} \xrightarrow{y \to 0} \psi^{[0]}$.

\medskip

\noindent
$\bullet$
{\em Uniform convergence of $F$ and relation of $F^\pm$:}
The claim of uniform convergence of \eqref{eqn:Sokhotski-uniform} and formula \eqref{eqn:Sokhotski}  now both follow from \eqref{82gfs}.

\medskip

\noindent
$\bullet$
{\em Boundedness of $F^\pm$:}
It is enough to provide a bound for $\psi^{[0]}(x)$. This can be achieved as follows. Split the integral:
\begin{align}
\int_{-L}^L \frac{\varphi(x,t)-\varphi(x,x)}{t-x}\;dt 
&
~=~ \int_{B_1(x)} \frac{\varphi(x,t)-\varphi(x,x)}{t-x}\;dt 
~+~
\int_{[-L,L]\setminus B_1(x)} \frac{\varphi(x,t)}{t-x}\; dt 
\nonumber
\\
&
\qquad
-~
\varphi(x,x)\int_{[-L,L]\setminus B_1(x)} \frac{1}{t-x}\; dt
\end{align}
The first summand can be estimated using the H\"older inequality (condition \ref{Sokhotski:2})
\be
\left|\int_{B_1(x)} \frac{\varphi(x,t)-\varphi(x,x)}{t-x}\;dt\right|\leq C\int_{-1}^1 \left|r\right|^{\alpha-1}\;dr = \frac{2C}{\alpha} 
\ee
For the second integral,
we make use of both boundedness (condition \ref{Sokhotski:6}, where as above we denote $S:=\sup_{(x,t)\in\mathbb{R}^2}|\varphi(x,t)|$)
and the decay property (condition \ref{Sokhotski:3}, w.l.o.g. $1<T<L$):
\begin{align}\left|\int_{[-L,L]\setminus B_1(x)} \frac{\varphi(x,t)}{t-x}\; dt\right| &\leq \int_{B_T(x)\setminus B_1(x)} \left|\frac{\varphi(x,t)}{t-x}\right|\; dt + \int_{[-L,L]\setminus B_T(x)} \left|\frac{\varphi(x,t)}{t-x}\right|\; dt \nonumber\\ 
&\leq \int_{B_T(x)\setminus B_1(x)} \left|\varphi(x,t)\right|\; dt + \int_{[-L,L]\setminus B_T(x)} \left|t-x\right|^{-\mu-1}\; dt \nonumber\\
&\leq 2TS + 2\int_1^\infty r^{-\mu-1}\; dr = 2TS+\frac{1}{\mu}
\end{align}
By \eqref{adi89},
the third integral is simply bounded as follows:
\be
\left|-\varphi(x,x)\int_{[-L,L]\setminus B_1(x)} \frac{1}{t-x}\; dt\right| 
\leq S\left|\log\left(\frac{L-x}{L+x}\right)\right| 
\ee
Altogether, we obtain the bound
\be
\left|\int_{-L}^L \frac{\varphi(x,t)-\varphi(x,x)}{t-x}\;dt\right| \leq \frac{2C}{\alpha} +\left(2T+\left|\log\left(\frac{L-x}{L+x}\right)\right|\right)S+\frac{1}{\mu} \ .
\ee
Since this bound itself converges as $L\rightarrow\infty$, we obtain a bound for $\left|\psi^{[0]}(x)\right|$.
\end{proof}

Write $D^+ := \{z\in D | \mathrm{Im}(z) > 0 \}$, $D^- := \{ z\in D | \mathrm{Im}(z) < 0 \}$ and $\tilde D^\pm := D^\pm \cup \mathbb{R}$.
Consider the functions
\be
	\tilde F^\pm : \tilde D^\pm \to \mathbb{C}
	\quad , \quad
	\tilde F^\pm(z) := 
	\begin{cases} 
	F(z) &;~ \mathrm{Im}(z) \neq 0
	\\
	F^\pm( z ) &;~ \mathrm{Im}(z) = 0
	\end{cases} \quad .
\ee

\begin{corollary} \label{corollarySokhotski} 
$\tilde F^\pm$ is a continuous extension of $F$ from 
 $D^\pm$ to $\tilde D^\pm$.
\end{corollary}

\begin{proof}
It suffices to show that $\psi(z)$ is continuous in $D$. We know already that it is continuous in $D\setminus \mathbb{R}$. Now let us show that it is continuous in $x_0\in\mathbb{R}$.

Let $\eps>0$.
By uniform convergence $\psi^{[y]}(x)\rightarrow\psi^{[0]}(x)$ there is a $\delta_1>0$ such that for all $|y|<\delta_1$ and all $x \in \mathbb{R}$ we have $|\psi^{[y]}(x)-\psi^{[0]}(x)|<\eps/2$.
By continuity of $F^\pm$ on $\mathbb{R}$ there is $\delta_2>0$ such that for all $x$ with $|x-x_0| < \delta_2$ we have $|\psi^{[0]}(x)-\psi^{[0]}(x_0)|<\eps/2$.

Take $\delta := \min(\delta_1,\delta_2)$. For $z = x + iy \in D$ with $|z-x_0|<\delta$ we have
\be
\left|\psi(z)-\psi(x_0)\right|=\left|\psi^{[y]}(x)-\psi^{[0]}(x_0)\right|\leq|\psi^{[y]}(x)-\psi^{[0]}(x)|+|\psi^{[0]}(x)-\psi^{[0]}(x_0)|\leq \varepsilon. \qedhere \ee
\end{proof}

\subsection{Proof of Proposition~\ref{SokhotskiConvolutionMain}} \label{AppendixPropProof}

Here we prove Proposition~\ref{SokhotskiConvolutionMain} as a special case of Proposition~\ref{Sokhotski} and Corollary~\ref{corollarySokhotski}.
Recall the setting of Proposition~\ref{SokhotskiConvolutionMain}:
\begin{itemize}
\item $D=\mathbb{S}_a$ for some $a>0$.
\item $\varphi(z,t)=h(z-t)g(t)$ where $h:\mathbb{S}_a\rightarrow\mathbb{C}$ is analytic and $g:\mathbb{R}\rightarrow\mathbb{C}$ is bounded and H\"older continuous.
\item $zh(z)$ and $\frac{d}{dz}h(z)$ are bounded in $\mathbb{S}_a$.
\end{itemize}
Let us check the conditions on $\varphi(z,t)$ one by one.
Denote by $G$ the bound of $g(t)$ and by $H$ the bound of $zh(z)$. 

\paragraph{Condition \ref{Sokhotski:1}} 
This is obvious since $h(z)$ is analytic.

\paragraph{Condition \ref{Sokhotski:2}} 
Let $C,\alpha$ be constants expressing the H\"older-continuity of $g$:
\be
\forall t,t' \in \Rb ~:~
|g(t')-g(t)| \le C |t-t'|^\alpha \ .
\ee
Boundedness of $\frac{d}{dz}h(z)$ implies that in particular the derivative $\frac{d}{dx}h(x+iy)$ in the real direction is bounded for every $y<(-a,a)$. As a consequence, for any given $z_0\in\mathbb{S}_a$ the function $t\mapsto h(z_0-t)$ is Lipschitz-continuous with Lipschitz constant $L := \mathrm{sup}_{z \in \mathbb{S}_a}(\frac{d}{dz}h(z))$.
Since $z h(z)$ is bounded, we also have boundedness of $h$ on $\mathbb{S}_a$: $|h(z)|\le B$ for some $B\ge 0$.
Accordingly,
\begin{align}
|\varphi(z_0,t)-\varphi(z_0,t')|
&=
|h(z_0-t)g(t)-h(z_0-t')g(t')|
\nonumber\\
&\le
|h(z_0-t)-h(z_0-t')|\,|g(t)|
+
|h(z_0-t')|\,|g(t)-g(t')|
\nonumber\\
&\le
  BC\, |t-t'|^\alpha
+ \begin{cases}
GL \, |t-t'| &;~ |t-t'| \le 1
\\
2BG &; ~ |t-t'|>1
\end{cases}
\quad .
\end{align}
Note that $B,G,L$ are all independent of $z_0$. 
Hence there is an $M>0$, independent of $z_0$, such that for all $t,t' \in \mathbb{R}$: 
\be
|\varphi(z_0,t)-\varphi(z_0,t')|
 \le M |t-t'|^\alpha \ .
\ee

\paragraph{Condition \ref{Sokhotski:3}}
For all $z_0\in \mathbb{S}_a$, we have the inequality 
\be
|\varphi(z_0,t)|\leq G|h(z_0-t)| \leq 
GH\,|z_0-t|^{-1} \ .
\ee
Fix $0<\mu<1$. 
There exists a $T>0$, independent of $z_0$, such that $GH\,|z_0-t|^{-1}\leq |z_0-t|^{-\mu}$ whenever $|z_0-t|>T$.

\paragraph{Condition \ref{Sokhotski:4}}
For given $z_0\in D\setminus \mathbb{R}$, set $\rho:=\frac{1}{2}\mathrm{Im}(z_0)$ and $U := \mathbb{S}_a \cap B_\rho(z_0)$, where $B_\rho(z_0)$ is the open ball of radius $\rho$ with center $z_0$.
Then for all $z\in U$, we have 
\be
|\varphi(z,t)| \leq G|h(z-t)| \leq GH\frac{1}{|z-t|}
	=
|z-t|\frac{GH}{|z-t|^2} \ . 
\ee
Now set
\be
M(t):=\begin{cases} \frac{GH}{|\mathrm{Re}(z_0)-\rho-t|^2} & \mathrm{if}\; t<\mathrm{Re}(z_0)-2\rho \\ \frac{GH}{\rho^2}& \mathrm{if}\; t\in [\mathrm{Re}(z_0)-2\rho,\mathrm{Re}(z_0)+2\rho] \\ \frac{GH}{|\mathrm{Re}(z_0)+\rho-t|^2} & \mathrm{if}\; t>\mathrm{Re}(z_0)+2\rho\end{cases} \ .
\ee
Then $M\in L_1(\mathbb{R})$ and one quickly checks that
\be
\frac{GH}{|z-t|^2}\leq M(t) \ .
\ee

\paragraph{Condition \ref{Sokhotski:5}}
Since $\varphi(x\pm iy,t)=h(x\pm iy -t)g(t)$, the condition is satisfied if $h(x+iy)\xrightarrow{y\to 0}h(x)$ uniformly. This is easily established with the following lemma.

\begin{lemma} \label{UniformConvergenceOfAnalyticFuns} Let $a>0$, $b\in\mathbb{R}$, and set $D:=\mathbb{S}_a\cap\lbrace z\in\mathbb{C}|\mathrm{Re}(z)> b\rbrace$. Let $f:D\rightarrow \mathbb{C}$ be an analytic function such that $zf(z)$ and $\frac{d}{dz}f(z)$ are bounded. Then $f(x+iy)\xrightarrow{y\to Y}f(x+iY)$ uniformly on $[t,\infty)$ for any $Y\in(-a,a)$ and any $t>b$.
\end{lemma}

\begin{proof}
Pointwise convergence is clear by continuity. Now we claim that the convergence is uniform. Let $\varepsilon>0$. Set $x_0:=\frac{2B}{\varepsilon}$, where $B>0$ is the bound of $zf(z)$. Without loss of generality, assume $b<t<x_0$. Then for all $x\geq x_0$ one has \begin{align*}\left|f(x+iy)-f(x+iY)\right| 
&\leq \left|f(x+iy)\right|+\left|f(x+iY)\right| \\
&\leq \frac{B}{|x+iy|}+\frac{B}{|x+iY|} \leq \frac{2B}{|x|} \leq\frac{2B}{x_0} =\varepsilon
\end{align*} for all $|y|<a$.

Thus, it remains to be shown that convergence on the compact interval $[t,x_0]$ is uniform. Since $f(z)$ is bounded, the family of functions $f^{[y]}(x):=f(x+iy)$ on the interval $[t,x_0]$ is uniformly bounded. Moreover, boundedness of $\frac{d}{dz}f(z)$ means that $\frac{d}{dx}f^{[y]}(x)$ are uniformly bounded. But this implies that $f^{[y]}(x)$ are equicontinuous. Thus, we can apply the Arzela-Ascoli theorem: for every sequence $y_n\rightarrow Y$, the sequence of functions $f^{[y_n]}(x)$ has a uniformly convergent subsequence. Now assume that $f^{[y]}(x)$ do not converge uniformly on $[t,x_0]$. 
Then there exists a sequence $u_n\rightarrow Y$ and a sequence $x_n$ of points in $[t,x_0]$ such that $\left|f^{[u_n]}(x_n)-f(x_n)\right| \geq\varepsilon$ for all $n$. But then $f^{[u_n]}(x)$ has no uniformly convergent subsequence, which is a contradiction.
\end{proof}

\paragraph{Condition \ref{Sokhotski:6}} This is again obvious, since both $h(z)$ and $g(t)$ are bounded.

\medskip

This completes the proof of Proposition~\ref{SokhotskiConvolutionMain}.

\newcommand\arxiv[2]      {\href{http://arXiv.org/abs/#1}{#2}}
\newcommand\doi[2]        {\href{http://dx.doi.org/#1}{#2}}
\newcommand\httpurl[2]    {\href{http://#1}{#2}}

\end{document}